\documentclass[journal]{IEEEtran}

\pdfoutput=1
\usepackage[utf8]{inputenc}
\usepackage[pdftex]{graphicx}
\graphicspath{{./Figures/}{../Figures/}} 
\usepackage{epstopdf} 
\usepackage{float} 
\usepackage{subcaption} 
\usepackage{booktabs}
\usepackage{amsmath, amssymb} 
\usepackage{amsthm} 
\usepackage{multirow}
\usepackage{cases}
\usepackage{tocloft} 
\usepackage{enumitem} 
\usepackage{threeparttable}
\usepackage{tablefootnote}
\usepackage{multirow}

\usepackage{amsthm} 
    \newtheorem{lemma}{Lemma}
    
    \newtheorem{definition}{Definition}
    
    \newtheorem{remark}{Remark}
    \newtheorem{theorem}{Theorem}
    \newtheorem{corollary}{Corollary}
    
    \newtheorem{proposition}{Proposition}
\renewcommand{\qedsymbol}{$\blacksquare$}

\usepackage{xcolor}
\usepackage{hyperref}
\hypersetup{%
    colorlinks=true,
    linkcolor={blue!50!black},
    citecolor={blue!50!black},
    urlcolor={blue!50!black}
}

\usepackage{dblfloatfix}

\usepackage[ruled, lined, longend, linesnumbered]{algorithm2e}
\usepackage[hang,flushmargin]{footmisc}
\usepackage{algpseudocode}
\SetKwInput{Input}{Input}
\SetKwInput{Declares}{Declares}
\SetKwInput{Require}{Require}
\SetKwInput{Prototype}{Prototype}
\SetKw{KwBy}{by}
\SetKw{Kwor}{or}
\SetKw{Kwand}{and}
\SetKw{KwEndFor}{end for}

\usepackage{fancyhdr}

\newcommand{\cc}[1]{{\mathcal{#1}}} 
\def\set#1#2{\left\{  #1 : #2 \right\}}
\def\R{\mathbb{R}} 
\def\Z{\mathbb{Z}} 
\def\Dp#1{\mathbb{D}_{+}^{#1}} 
\def\Sp#1{\mathbb{S}_{+}^{#1}} 

\newcommand{\T}{^\top} 
\def\sp#1#2{\langle #1,#2\rangle } 
\def\Sum#1#2{\sum\limits_{#1}^{#2}} 

\def\vv#1{{ \rm \bf{#1}}} 
\newcommand{\becauseof}[2][=]{\stackrel{\scriptstyle\mkern-1.5mu#2\mkern-1.5mu}{#1}}
\def\moveEq#1{{}\mkern#1mu} 
\def\diag{\mathtt{diag}}

\newcommand{\yLB}{\underline{y}}
\newcommand{\yLBj}{\underline{y}_{(i)}}
\newcommand{\yUB}{\overline{y}}
\newcommand{\yUBj}{\overline{y}_{(i)}}
\def\xe{x_e} 
\def\xs{x_s}
\def\xc{x_c}
\def\ue{u_e} 
\def\us{u_s}
\def\uc{u_c}

\def\xr{x_r}
\def\ur{u_r}
\def\xre{x_{re}} 
\def\xrs{x_{rs}}
\def\xrc{x_{rc}}
\def\ure{u_{re}} 
\def\urs{u_{rs}}
\def\urc{u_{rc}}

\def\xh{x_h} 
\def\xhj{x_{h}^{k}} 

\def\xH{\vv{x}_h} 
\def\uh{u_h} 
\def\uhj{u_{h}^{k}} 

\def\uH{\vv{u}_h} 
\def\xHo{\mathring{\vv{x}}}
\def\uHo{\mathring{\vv{u}}}
\def\xeo{\mathring{x}_e}
\def\ueo{\mathring{u}_e}
\def\xso{\mathring{x}_s}
\def\uso{\mathring{u}_s}
\def\xco{\mathring{x}_c}
\def\uco{\mathring{u}_c}
\def\xrH{\vv{x}_{r}} 
\def\urH{\vv{u}_{r}} 
\def\xho{\mathring{x}}
\def\uho{\mathring{u}}

\newcommand{\Tw}[2][w]{\mathcal{T}_{#1}^{#2}}
\newcommand{\Twa}[1]{\mathcal{T}_h^{#1}}
\def\cD{\cc{D}}
\def\cC{\cc{C}_\sigma}

\begin{document}

\title{\LARGE Harmonic model predictive control for tracking sinusoidal references and its application to trajectory tracking}

\author{Pablo~Krupa$^*$,~Daniel~Limon$^\dagger$,~Alberto~Bemporad$^*$,~Teodoro~Alamo$^\dagger$%
    \thanks{$^*$ IMT School for Advanced Studies, Piazza San Francesco 19, Lucca, Italy. Emails: \texttt{\{pablo.krupa, alberto.bemporad\}@imtlucca.it}}%
\thanks{$^\dagger$ Department of Systems Engineering and Automation, Universidad de Sevilla, Seville, Spain. E-mails: \texttt{dlm@us.es}, \texttt{talamo@us.es}.}%
\thanks{Corresponding author: Pablo Krupa.}%
\thanks{This work has been funded by grant PID2022-141159OB-I00 funded by MCIN/AEI/10.13039/501100011033 and by ERDF/EU.}
}

\pagestyle{fancy}
\maketitle
\thispagestyle{fancy}

\begin{abstract}
Harmonic model predictive control (HMPC) is a recent model predictive control (MPC) formulation for tracking piece-wise constant references that includes a parameterized artificial harmonic reference as a decision variable, resulting in an increased performance and domain of attraction with respect to other MPC formulations.
This article presents an extension of the HMPC formulation to track periodic harmonic/sinusoidal references and discusses its use for tracking arbitrary trajectories.
The proposed formulation inherits the benefits of its predecessor, namely its good performance and large domain of attraction when using small prediction horizons, and that the complexity of its optimization problem does not depend on the period of the reference.
We show closed-loop results discussing its performance and comparing it to other MPC formulations.
\end{abstract}


\section{Introduction} \label{sec:intro}

The use of model predictive control (MPC) \cite{rawlings_model_2017} to track periodic references is a widely studied problem in the control literature, since it has many practical applications, such as repetitive control \cite{gupta_period-robust_2006}, control of periodic systems \cite{leomanni_sum--norms_2020, gondhalekar_mpc_2011}, or economic MPC \cite{risbeck_economic_2020}. 
In \cite{limon_mpc_2016}, a linear MPC for tracking periodic references was presented as an extension of the MPC for tracking piece-wise affine references from \cite{limon_mpc_2008}.
This formulation makes use of an artificial periodic reference trajectory, which becomes part of its optimization problem.
The benefits of using this artificial reference are that the resulting MPC controller is recursively feasible even in the event of a reference change, and that the closed-loop system converges to the periodic trajectory that is ``closest'' to the periodic reference, where the distance is measured by the terminal cost of the MPC controller. Thus, the formulation inherently deals with references that cannot be perfectly tracked.
Additionally, the use of the artificial reference results in a domain of attraction that is typically significantly larger than the ones obtained from classical MPC formulations \cite{ferramosca_mpc_2009}.
The disadvantage of this approach is that the number of decision variables of the artificial periodic reference grows with its period, thus increasing the complexity of the optimization problem.
The use of artificial references in MPC for tracking periodic references has been extended to nonlinear MPC \cite{Kohler_NMPC_18, kohler_AUT_2020, yang_nonlinear_2021}, as well as applied to economic MPC~\cite{Limon_JPC_2014, Pereira_PEMPC_2015} and distributed (periodic) MPC \cite{Kohler_IFAC_2023}.

In \cite{krupa_harmonic_2022}, the authors present a linear MPC formulation that uses a parameterized harmonic signal as an artificial reference for tracking piece-wise affine set-points.
The formulation, named Harmonic Model Predictive Control (HMPC), retains the recursive feasibility and asymptotic stability properties of the original linear MPC for tracking~\cite{limon_mpc_2008}.
Additionally, in~\cite{krupa_harmonic_2022}, the authors show that the use of a harmonic artificial reference may lead to a significantly larger domain of attraction and better performance when working with small prediction horizons.
The downside of the HMPC formulation is that the use of a harmonic artificial reference comes at the cost of the inclusion of second-order cone constraints, resulting in an optimization problem which is no longer a quadratic programming (QP) problem.
However, in \cite{krupa_efficiently_2022} the authors presented an efficient solver for the HMPC formulation, showing that its solution-time is comparable to state-of-the-art QP solvers applied to alternative MPC formulations.

Harmonic signals (otherwise known as sinusoidal signals) are a particular class of periodic signal that are found in many practical applications, such as power electronics \cite{Karamanakos_OJIA_2020, Ordonez_TPE_2023}, robotics~\cite{Lafmejani_RAL_2020} or spacecraft rendezvous \cite{leomanni_sum--norms_2020,dong_novel_2022}.

In this paper, we present an extension of HMPC for tracking harmonic references, instead of the piece-wise constant references considered in \cite{krupa_harmonic_2022}.
This is a natural extension, given the fact that the HMPC formulation uses a harmonic artificial reference and the wide range of applications where harmonic signals naturally occur.
The benefit of the proposed linear MPC formulation, when compared to periodic MPC, is that the complexity of its optimization problem does not depend on the period of the harmonic reference.
The formulation also inherently deals with non-admissible references, resulting in a closed-loop behavior that may differ from the typical one obtained from other periodic MPC formulations due to the terminal ingredients partly measuring the ``distance'' to the reference in terms of it ``shape'', as we show in the numerical case study.
The optimization problem of the proposed extension can be solved using a minor modification of the solver from~\cite{krupa_efficiently_2022} (available in \cite{krupa_spcies_2020}).
Additionally, the formulation retains the recursive feasibility and asymptotic stability of the original HMPC, as well as its good performance and large domain of attraction when using small prediction horizons.

The proposed formulation may have other useful applications, such as obstacle avoidance~\cite{Pereira_Obstacles_2021, dosSantos_AUT_2024} or tracking of generic reference trajectories~\cite{kohler_TAC_2020}, due to its large domain of attraction, the elliptic nature of its artificial reference, and its aforementioned ``shape-tracking'' behavior.
In particular, we find that the dynamic nature of the harmonic artificial reference leads to a remarkably good performance when applied to the problem of tracking generic reference trajectories, as we discuss in Section~\ref{sec:arbitrary:ref} and show in the numerical case study.

\subsubsection*{Notation}

Given two vectors $x$ and $y$, $x \leq (\geq) \; y$ denotes componentwise inequalities.
Given two integers $i$ and $j$ with ${j \geq i}$, $\Z_i^j$ denotes the set of integer numbers from $i$ to $j$, i.e. ${\Z_i^j \doteq \{i, i+1, \dots, j-1, j\}}$. 
We denote by $\Sp{n}$ ($\Dp{n}$) the set of (diagonal) positive definite matrices in $\R^{n \times n}$.
For vectors $x_1$ to $x_N$, $(x_{1}, x_{2}, \dots, x_{N})$ denotes the column vector formed by their concatenation.
Given a vector $x\in \R^{n}$, we denote its $i$-th component using a parenthesized subindex $x_{(i)}$.
Given two vectors $x \in \R^{n}$ and $y \in \R^{n}$, their standard inner product is denoted by $\sp{x}{y} \doteq \sum_{i=1}^{n} x_{(i)} y_{(i)}$.
For $x \in \R^{n}$ and $A \in \Sp{n}$, $\|x\| \doteq \sqrt{\sp{x}{x}}$ and $\|x\|_A \doteq \sqrt{\sp{x}{A x}}$.
The identity matrix of dimension $n$ is denoted by $I_n$.
Given scalars and/or matrices $M_1, \dots, M_N$, we denote by $\texttt{diag}(M_1, \dots, M_N)$ the block diagonal matrix formed by the concatenation of $M_1$ to $M_N$.

\section{Admissible harmonic signals} \label{sec:prelim:harmonic}

\begin{definition}[Harmonic signal] \label{def:harmonic}
A trajectory $v(\cdot) \in \R^{m}$ is a \textit{harmonic signal} if it satisfies
\begin{equation} \label{eq:harmonic}
    v(t) = v_e + v_s \sin(w t) + v_c \cos(w t), \\
\end{equation}
for some parameters $v_e, v_s, v_c \in \R^{m}$ and frequency $w > 0$.
\end{definition}

In the following, we use a bold $\vv{v} \doteq (v_e, v_s, v_c)$ as a shorthand to indicate the parameters of a harmonic signal~$v(\cdot)$, where their dimension is inferred from the dimension of~$v(\cdot)$.

Consider a system described by a controllable linear time-invariant state-space model
\begin{equation} \label{eq:model}
    x(t+1) = A x(t) + B u(t),
\end{equation}
where $x(t) \in \R^{n_x}$ and $u(t) \in \R^{n_u}$ are the state and control input at the discrete time instant $t$, respectively, subject to
\begin{equation} \label{eq:constraints}
    \yLB \leq E x(t) + F u(t) \leq \yUB,\;\forall t,
\end{equation}
where $E \in \R^{n_y \times n_x}$, $F \in \R^{n_y \times n_u}$, $\yLB, \yUB \in \R^{n_y}$ and $\yLB < \yUB$.

\begin{definition}[Admissible harmonic signals] \label{def:harmonic:admissible}
The harmonic signals $\hat{x}(\cdot)$ and $\hat{u}(\cdot)$ with frequency $w > 0$, parametrized by $\vv{\hat{x}} \doteq (\hat{x}_e, \hat{x}_s, \hat{x}_c)$ and $\vv{\hat{u}} \doteq (\hat{u}_e, \hat{u}_s, \hat{u}_c)$, are \textit{admissible} if they satisfy ${\hat{x}(t+1)} = A \hat{x}(t) + B \hat{u}(t)$ and $\yLB \leq E \hat{x}(t) + F \hat{u}(t) \leq \yUB$, $\forall t$, where if the inequalities are strictly satisfied $\forall t$, we say that they are \textit{strictly admissible}.
\end{definition}

The following propositions provide sufficient conditions for admissibility of a harmonic signal, where we use the notation
$\hat{y}_e \doteq E \hat{x}_e + F \hat{u}_e$, $\hat{y}_s \doteq E \hat{x}_s + F \hat{u}_s$, and $\hat{y}_c \doteq E \hat{x}_c + F \hat{u}_c$,
for clarity of presentation.

\begin{proposition} \label{prop:harmonic:dynamics}
Let $\hat{x}(\cdot)$ and $\hat{u}(\cdot)$ be harmonic signals with the same frequency $w$ parametrized by $ \vv{\hat{x}} \doteq (\hat{x}_e, \hat{x}_s, \hat{x}_c)$ and $\vv{\hat{u}} \doteq (\hat{u}_e, \hat{u}_s, \hat{u}_c)$, and 
\begin{equation*}
    \cD \doteq \set{(\vv{\hat{x}}, \vv{\hat{u}})}{ \begin{array}{@{}l@{}} \hat{x}_e = A \hat{x}_e + B \hat{u}_e \\ \hat{x}_s \cos(w) - \hat{x}_c \sin(w) = A \hat{x}_s + B \hat{u}_s\\ \hat{x}_s \sin(w) + \hat{x}_c \cos(w) = A \hat{x}_c + B \hat{u}_c \end{array}}.
\end{equation*}
Then, $(\vv{\hat{x}}, \vv{\hat{u}}) \in \cD$ implies $\hat{x}(t+1) = A \hat{x}(t) + B \hat{u}(t)$, $\forall t$.
\end{proposition}%
{\renewcommand{\qedsymbol}{}
\begin{proof}
The proposition follows from \cite[Property 2]{krupa_harmonic_2022}.
\end{proof}}

\begin{proposition} \label{prop:harmonic:constraints}
Let $\hat{x}(\cdot)$ and $\hat{u}(\cdot)$ be harmonic signals with the same frequency $w$ parametrized by $ \vv{\hat{x}} \doteq (\hat{x}_e, \hat{x}_s, \hat{x}_c)$ and $\vv{\hat{u}} \doteq (\hat{u}_e, \hat{u}_s, \hat{u}_c)$, and 
    \begin{equation*}
        \cC \doteq \set{(\hat{x}, \hat{u})}{ \begin{array}{@{}l@{}} (\hat{y}_{e(i)}, \hat{y}_{s(i)}, \hat{y}_{c(i)}) \in \overline{\cc{Y}}_i \cap \underline{\cc{Y}}_i, \, i \in \Z_1^{n_y} \end{array}},
    \end{equation*}
where sets $\overline{\cc{Y}}_i$ and $\underline{\cc{Y}}_i$ are defined as
\begin{equation*} \label{eq:Y:sets}
\begin{aligned}
    \overline{\cc{Y}}_i &\doteq \set{y {=} (y_0, y_1)}{ y_0 {\in} \R, y_1 {\in} \R^{2}, \| y_1\| \leq \yUBj - \sigma - y_0}, \\
    \underline{\cc{Y}}_i &\doteq \set{y {=} (y_0, y_1)}{y_0 {\in} \R, y_1 {\in} \R^{2}, \| y_1\| \leq y_0 - \yLBj - \sigma },
\end{aligned}
\end{equation*}
and $\sigma \geq 0$.
Then, $(\vv{\hat{x}}, \vv{\hat{u}}) \in \cC$ implies that $\hat{x}(\cdot)$ and $\hat{u}(\cdot)$ satisfy $\yLB \leq E \hat{x}(t) + F \hat{u}(t) \leq \yUB$, $\forall t$, where the implication follows with strict inequality if $\sigma > 0$.
\end{proposition}
{\renewcommand{\qedsymbol}{}
\begin{proof}
    The proposition directly follows from \cite[Property 3]{krupa_harmonic_2022} considering the definitions of $\hat{y}_e$, $\hat{y}_s$, $\hat{y}_c$, $\overline{\cc{Y}}_i$ and $\underline{\cc{Y}}_i$.
\end{proof}}

\begin{corollary} \label{coro:harmonic}
A harmonic signal $(\hat{x}(\cdot), \hat{u}(\cdot))$ satisfying the conditions of Propositions \ref{prop:harmonic:dynamics} and \ref{prop:harmonic:constraints} is an admissible harmonic signal of system \eqref{eq:model} subject to \eqref{eq:constraints}.
Furthermore, it is a strict admissible harmonic signal if the conditions of Proposition \ref{prop:harmonic:constraints} are satisfied for some $\sigma > 0$.
\end{corollary}

The key point of Propositions \ref{prop:harmonic:dynamics} and \ref{prop:harmonic:constraints} is that they provide conditions for admissibility of the harmonic signal $(\hat{x}(\cdot), \hat{u}(\cdot))$ that only depend on the parameters $\vv{\hat{x}}$ and $\vv{\hat{u}}$, but not on their frequency $w$.
Satisfaction of the state dynamics is guaranteed by the satisfaction of the linear constraints in $\cD$, whereas satisfaction of the system constraints is guaranteed by the satisfaction of the two second order cone constraints in~$\cC$.

\begin{figure}[t]
    \centering
    \includegraphics[width=0.85\linewidth]{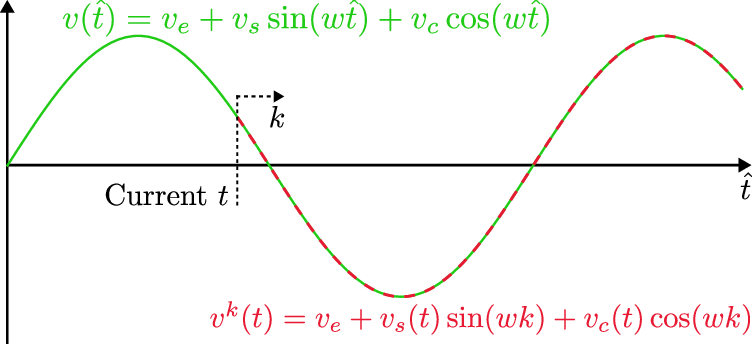}
    \caption{Harmonic signal expressed in the relative time $k$.}
    \label{fig:harmonic:rel}
\end{figure}

In the following, we will be interested in expressing harmonic signals in terms of a relative time $k$ with respect to the current time $t$.
That is, for a fixed time $t$, we want to obtain an expression \eqref{eq:harmonic} for $v^k(t) \doteq v(t+k) \in \R^m$ in the form
\begin{equation} \label{eq:rel:harmonic}
    v^k(t) = v_e(t) + v_s(t) \sin(w k) + v_c(t) \cos(w k)
\end{equation}
for some time-varying parameters $v_e(t), v_s(t), v_c(t) \in \R^m $.
Expression \eqref{eq:rel:harmonic} is simply a time-shift of the underlying harmonic signal $v(\cdot)$, as illustrated in Fig.~\ref{fig:harmonic:rel}, where the current time $t$ is taken as the ``initial time'' of the periodic signal $v^k(\cdot)$, i.e., $v^0(t) = v(t)$.
As shown in \cite[Property 1]{krupa_harmonic_2022}, the parameters $\vv{v}(t) \doteq (v_e(t), v_s(t), v_c(t))$ of \eqref{eq:rel:harmonic} are obtained from the recursion $v_e(0) = v_e$, $v_s(0) = v_s$, $v_c(0) = v_c$,
\begin{equation} \label{eq:rel:harmonic:update}
    \vv{v}(t+1) = \Twa{m} \vv{v}(t),\, \forall t \geq 0,
\end{equation}
where $\Twa{m} \doteq \diag(I_{m}, \Tw{m})$ and 
\begin{equation} \label{eq:def:Tw}
\Tw{m} \doteq \begin{bmatrix} I_m \cos(w) & -I_m \sin(w) \\ I_m \sin(w) & I_m \cos(w) \end{bmatrix} \in \R^{2 m \times 2 m}.
\end{equation}

\section{HMPC for harmonic reference tracking} \label{sec:HMPC}

The idea behind HMPC \cite{krupa_harmonic_2022} is to introduce an artificial harmonic reference in the optimization problem whose discrepancy with the desired reference is penalized in the objective function.
Additionally, the objective function penalizes the discrepancy between the predicted system trajectory and this artificial reference, as is typical in MPC.
In \cite{krupa_harmonic_2022}, HMPC was used to track set-point references, i.e., piecewise affine references $(x_r, u_r) \in \R^{n_x} \times \R^{n_u}$.
In this article, however, the control objective is to track a harmonic reference trajectory $(\xr(\cdot), \ur(\cdot))$, instead of a piecewise constant reference.
At each sample time $t$, this reference trajectory can be equivalently expressed by its relative-time signals
\begin{subequations} \label{eq:reference:k}
\begin{align}
    \xr^k(t) &= \xre + \xrs(t) \sin(w k) + \xrc(t) \cos(w k), \label{eq:xr:k}\\
    \ur^k(t) &= \ure + \urs(t) \sin(w k) + \urc(t) \cos(w k), \label{eq:ur:k}
\end{align}
\end{subequations}
with suitable values of $\xrH(t) \doteq (\xre, \xrs(t), \xrc(t))$ and $\urH(t) \doteq (\ure, \urs(t), \urc(t))$, as discussed in Section~\ref{sec:prelim:harmonic}.
Notice that we make no assumption on the admissibility of the reference.
If is is an admissible harmonic signal (see Definition~\ref{def:harmonic:admissible}) then we wish to converge to it.
Otherwise, we wish to converge to its closest admissible harmonic trajectory, for a criterion of proximity that will be apparent further ahead.

HMPC, as is typical in MPC, uses the notion of receding horizon, where at each sample time $t$ we consider a window of future predictions indexed by $k \in \Z_0^N$, where $N > 0$ is the prediction horizon and $k = 0$ corresponds to the current time~$t$.
The artificial harmonic reference has the same form as the reference \eqref{eq:reference:k}, i.e., sequences $\{\xhj\} \in \R^{n_x}$ and $\{\uhj\} \in \R^{n_u}$ whose values at each prediction time $k \in \Z_0^N$ are given by
\begin{subequations} \label{eq:harmonic:signals}
\begin{align}
    \xhj &= \xe + \xs \sin(w k) + \xc \cos(w k), \label{eq:x_h}\\
    \uhj &= \ue + \us \sin(w k) + \uc \cos(w k), \label{eq:u_h}
\end{align}
\end{subequations}
where the parameters $\xH \doteq (\xe, \xs, \xc)$ and $\uH \doteq (\ue, \us, \uc)$ are decision variables of the HMPC's optimization problem, which, at each sample time $t$ and for a given choice of the prediction horizon $N$ and a reference parameterized by $\xrH(t)$ and $\urH(t)$, is given by
\begin{subequations} \label{eq:HMPC} 
\begin{align}
    \moveEq{-14}\min\limits_{\substack{\vv{x},\vv{u},\\ \xH, \uH}} \;& V_h(\xH, \uH; \xrH(t), \urH(t)) + \Sum{k = 0}{N-1} \ell_h(x^k, u^k, \xhj, \uhj)  \label{eq:HMPC:cost} \\
    {\rm s.t.} \;& x^0 = x(t) \label{eq:HMPC:cond:inic}\\
    &x^{k+1} = A x^k + B u^k, \; k\in\Z_0^{N-1} \label{eq:HMPC:dynamics}\\
    & \yLB \leq E x^k + F u^k \leq \yUB, \; k\in\Z_0^{N-1} \label{eq:HMPC:constraints}\\
    & x^N = \xe + \xs \sin(w N) + \xc \cos(w N) \label{eq:HMPC:xN}\\
    & (\xH, \uH) \in \cD \label{eq:HMPC:D} \\
    & (\xH, \uH) \in \cC, \label{eq:HMPC:C}
\end{align}
\end{subequations}
where $\vv{x} = (x^0, \dots, x^{N-1} )$, $\vv{u} = ( u^0, \dots, u^{N-1} )$, the two terms of the cost function are given by the stage cost function
\begin{equation*}
    \ell_h(x, u, \xh, \uh) = \| x - \xh \|_Q^2 + \| u - \uh \|_R^2,
\end{equation*}
with $Q\in\Sp{n_x}$, $R\in\Sp{n_u}$, and the offset cost function 
\begin{align*}
    V_h(\cdot) &= \| \xe - \xre \|_{T_e}^2 + \| \xs - \xrs(t) \|_{T_h}^2 + \| \xc - \xrc(t) \|_{T_h}^2 \nonumber \\
               & + \| \ue - \ure \|_{S_e}^2 + \| \us - \urs(t) \|_{S_h}^2 + \| \uc - \urc(t) \|_{S_h}^2,
\end{align*}
with $T_e\in\Sp{n_x}$, $T_h \in\Dp{n_x}$, $S_e\in\Sp{n_u}$, and $S_h \in\Dp{n_u}$; and $\sigma > 0$ is taken as an arbitrarily small scalar to avoid a possible controllability loss in the presence of active constraints at an equilibrium point \cite{limon_mpc_2008}.

Constraints \eqref{eq:HMPC:cond:inic}-\eqref{eq:HMPC:constraints} impose the typical MPC constraints, namely, the initial state, system dynamics and system constraints.
Constraint \eqref{eq:HMPC:xN} imposes that the predicted state $x^N$ reaches the value of the artificial harmonic reference at $k = N$.
The equality constraints \eqref{eq:HMPC:D} impose the satisfaction of the system dynamics \eqref{eq:model} on the artificial harmonic reference, as shown in Proposition \ref{prop:harmonic:dynamics}, whereas the second order cone constraints \eqref{eq:HMPC:C} impose the strict satisfaction of the system constraints \eqref{eq:constraints} on the artificial harmonic reference, as shown in Proposition \ref{prop:harmonic:constraints}.
The satisfaction of \eqref{eq:HMPC:D} and \eqref{eq:HMPC:C} implies that the artificial harmonic reference is a strictly admissible harmonic signal of system \eqref{eq:model} subject to \eqref{eq:constraints}, where strict satisfaction of the constraints is attained due to the inclusion of the scalar $\sigma > 0$ in \eqref{eq:Y:sets}, as stated in Corollary \ref{coro:harmonic}.

The cost function of the HMPC formulation penalizes, on one hand, the discrepancy between the predicted states $x^k$ and inputs $u^k$ with the values $\xhj$ and $\uhj$ of the artificial harmonic reference at prediction time $k$, respectively, and on the other hand, the discrepancy between the parameters $(\xH, \uH)$ with $(\xrH(t), \urH(t))$.
The effect is that the artificial harmonic reference will tend towards the reference, while in turn the predicted states will tend towards the artificial harmonic reference.
Let $\vv{x}^*$, $\vv{u}^*$, $\xH^*$ and $\uH^*$ be the optimal solution of \eqref{eq:HMPC}.
The control input $u(t)$ is taken as the first move in the sequence of optimal inputs $\vv{u}^*$.
We note that~\eqref{eq:HMPC} can recover the classical local optimality of MPC for tracking (see \cite[Property~1]{ferramosca_mpc_2009}) if, for instance, linear penalization terms are added to $V_h$, i.e, $\alpha_e \| \xe - \xre \|$, $\alpha_s \| \xs - \xrs(t) \|$, $\alpha_c \| \xc - \xrc(t) \|$, $\beta_e \| \ue - \ure \|$, $\beta_s \| \us - \urs(t) \|$, $\beta_c \| \uc - \urc(t) \|$, with $\alpha_{\{e, s, c\}}, \beta_{\{e, s, c\}} \geq 0$ large enough.
However, in this article we focus on the previously shown quadratic offset cost $V_h$, since it results in a much simpler to solve optimization problem; a typicall choice in linear MPC for tracking for this reason~\cite{krupa_MPCT_implem_2021}.

\subsection{Properties of the HMPC formulation \eqref{eq:HMPC}} \label{sec:HMPC:property}

In this subsection we formally establish the convergence, stability and recursive feasibility guarantees of the HMPC formulation \eqref{eq:HMPC}.
We start by presenting a key concept of the formulation that we label the \textit{optimal reachable harmonic reference}, which plays a mayor role in the convergence results.

\begin{definition}[Optimal reachable harmonic reference]
At sample time $t$, we define the \textit{optimal reachable harmonic reference sequence} of the HMPC formulation \eqref{eq:HMPC} for the given reference $(\xr(t), \ur(t))$ as the harmonic sequences $\{\xho^k\}$, $\{\uho^k\}$ parameterized by the unique solution $\xHo(t) = (\xeo(t), \xso(t), \xco(t))$ and $\uHo(t) = (\ueo(t), \uso(t), \uco(t))$ of
\begin{subequations} \label{eq:xHo:OP}
\begin{align}
    (\xHo(t), \uHo(t)) = \arg\min\limits_{\xH, \uH}\;& V_h(\xH, \uH, \xrH(t), \urH(t)) \\
    {\rm s.t.} \;& (\xH, \uH) \in \cD \cap \cC.
\end{align}
\end{subequations}
\end{definition}

The following lemma establishes the relation between the parameters characterizing the optimal reachable harmonic sequences of two consecutive time instants.
Its proof can be found in the appendix.

\begin{lemma} \label{lemma:xH}
    Assume that $\xrH(t+1) = \Twa{n_x} \xrH(t)$ and $\urH(t+1) = \Twa{n_u} \urH(t)$. Then, $\xHo(t+1) = \Twa{n_x} \xHo(t)$, $\uHo(t+1) = \Twa{n_u} \uHo(t), \forall t$.
\end{lemma}

The main consequence of Lemma \ref{lemma:xH} is that the optimal reachable harmonic reference is in fact a unique trajectory $(\xho(\cdot), \uho(\cdot))$ for each harmonic reference trajectory $(\xr(\cdot), \ur(\cdot))$.
We formalize this in the following corollary.

\begin{corollary} \label{cor:xH}
The optimal reachable harmonic reference sequences obtained at subsequent times $t$ define a unique trajectory $(\xho(\cdot), \uho(\cdot))$ given by
\begin{align*}
    \xho(t) &= \xeo + \xso \sin(w t) + \xco \cos (w t), \\
    \uho(t) &= \ueo + \uso \sin(w t) + \uco \cos (w t),
\end{align*}
where $\xeo = \xeo(0)$, $\xso = \xso(0)$, $\xco = \xco(0)$, $\ueo = \ueo(0)$, $\uso = \uso(0)$ and $\uco = \uco(0)$.
\end{corollary}

The following theorems state the recursive feasibility and asymptotic stability of the HMPC formulation \eqref{eq:HMPC}, where recursive stability is maintained even if the reference trajectory $(\xr(\cdot), \ur(\cdot))$ is changed between sample times and asymptotic stability is satisfied with respect to the optimal reachable harmonic reference trajectory $\xho(\cdot)$.
As a results, it is easy to verify from~\eqref{eq:xHo:OP} that the system will converge, under nominal conditions, to the desired reference trajectory $\xr(\cdot)$ if $(\xr(\cdot), \ur(\cdot))$ is a strictly admissible harmonic trajectory, i.e., if $(\xrH(0), \urH(0)) \in \cD \cap \cC$, where $t = 0$ for convenience, since any value of $t$ can be taken\footnote{We note that $\sigma > 0$ is included for technical reasons, but can be chosen as an arbitrary small number. Thus, as long as the reference trajectory does not reach an active constraint~\eqref{eq:constraints}, a sufficiently small $\sigma$ can always be selected.}.
If the reference is not admissible, then the system will asymptotically converge to the ``closest" harmonic trajectory to $(\xr(\cdot), \ur(\cdot))$, where ``closeness" is defined in terms of the offset cost function $V_h(\cdot)$.
We refer the reader to the appendix for sketches of the proofs of the following theorems, and to preprint~\cite{krupa_ellipHMPC_arXiv_2023_v1} for detailed proofs.

\begin{theorem}[Recursive feasibility] \label{theo:feasibility}
    Suppose that $x(t)$ belongs to the feasibility region of the HMPC controller \eqref{eq:HMPC} for some fixed choice of $N > 0$ and that $\vv{x}$, $\vv{u}$, $\xH$, $\uH$ constitute any feasible solution of \eqref{eq:HMPC} for $x(t)$ and the reference $(\xrH(t), \urH(t))$.
    Then, the successor state ${x(t+1) = A x(t) + B u^0}$ belongs to the feasibility region of \eqref{eq:HMPC} for any $(\tilde{\vv{x}}_r(t+1), \tilde{\vv{u}}_r(t+1))$ not necessarily equal to $(\Twa{n_x}\xrH(t), \Twa{n_u}\urH(t))$.
\end{theorem}

\begin{theorem}[Asymptotic stability] \label{theo:stability}
Consider a controllable system \eqref{eq:model} subject to \eqref{eq:constraints} controlled with the HMPC formulation \eqref{eq:HMPC} with $N$ greater or equal to the controllability index of the system.
Then, for any given harmonic reference trajectory $(\xr(\cdot), \ur(\cdot))$ and initial state $x(0)$ belonging to the feasibility region of the HMPC formulation \eqref{eq:HMPC}, the closed-loop system trajectory $x(\cdot)$ is stable, satisfies the system constraints for all $t$, and asymptotically converges\footnote{That is, $\| x(t) - \xho(t) \| \to 0$ as $t \to \infty$, c.f.~\cite[Appendix~B.2]{rawlings_model_2017}.}~to the~optimal reachable harmonic reference trajectory
$\xho(\cdot)$ given by Corollary \ref{cor:xH}.
That is, there exists a $\mathcal{K}\mathcal{L}$ function $\beta(\cdot)$ satisfying $\| x(t) - \xho(t) \| \leq \beta(\| x(0) - \xho(0) \|, t)$, $\forall t \geq 0$, c.f.~\cite[Theorem.~B.15]{rawlings_model_2017}.
\end{theorem}

\begin{remark} \label{rem:shape}
An interesting consequence of the parametrization of the reference is the effect it has on the optimal reachable harmonic reference, i.e., on the reference to which the closed-loop system converges to, particularly when the reference $(\xr(\cdot), \ur(\cdot))$ is non-admissible.
The terms $\| \xe - \xre \|_{T_e}^2$ and $\| \ue - \ure \|_{S_e}^2$ of the offset cost function $V_h$ penalize the distance between the ``centers'' of both references, whereas the other terms penalize the discrepancy between the parameters that characterize the sine and cosine terms, which intuitively can be seen as a penalization on the discrepancy between the ``shapes'' of the references.
Therefore, if $T_h$ and $S_h$ are significantly larger that $T_e$ and $S_e$, the closed-loop system will converge to the harmonic trajectory that is closest to the given reference but that retains its shape, as shown in Section~\ref{sec:results:harmonic}.
\end{remark}

\subsection{Numerically solving the HMPC formulation \eqref{eq:HMPC}} \label{sec:HMPC:solving}

In \cite{krupa_efficiently_2022}, the authors present a method for efficiently solving the HMPC formulation for set-point tracking from \cite{krupa_harmonic_2022} that is applied to the alternating direction method of multipliers (ADMM) algorithm \cite{boyd_distributed_2010} to obtain a sparse solver that is available in the open-source Matlab toolbox SPCIES \cite{krupa_spcies_2020}.
The results in \cite{krupa_efficiently_2022} show that the HMPC formulation can be solved in times comparable to other MPC formulations using state-of-the-art solvers.
The same ADMM-based solver can be applied to \eqref{eq:HMPC} by making very minor changes, since the reference only affects a submatrix of the Hessian of \eqref{eq:HMPC}.
In fact, the solver for \eqref{eq:HMPC} is also available in \cite[\texttt{v0.3.11}]{krupa_spcies_2020}.

\begin{remark}
Note that the information of the reference is provided to the HMPC formulation with the parameters $\xrH(t)$ and $\urH(t)$, which are independent of the value of its period (determined by $w$).
This, along with the way in which the system dynamics and constraints are imposed on the artificial harmonic reference, i.e., by means of \eqref{eq:HMPC:D} and \eqref{eq:HMPC:C}, leads to an optimization problem whose complexity does not depend on the value of $w$.
This is not the case in other periodic MPC formulations \cite{limon_mpc_2016}, where the number of constraints grows with the period of the reference trajectory, resulting in a increase of the computational complexity of the solver.
\end{remark}

\section{Tracking arbitrary references} \label{sec:arbitrary:ref}

In this section we discuss the application of HMPC~\eqref{eq:HMPC} to tracking generic references, i.e., we no longer assume that the reference $(\xr(\cdot), \ur(\cdot))$ describes a harmonic signal (see Definition~\ref{def:harmonic}), although we do assume that is satisfies the system dynamics, i.e., $\xr(t+1) = A \xr(t) + B \ur(t)$, $\forall t$.

In general, HMPC will not be able to track the reference $\xr(\cdot)$ due to the use of a single-harmonic artificial reference, i.e., a harmonic signal with a single frequency $w$.
However, we find that it is often able to track a suitably selected output $y_r = C \xr$.
To do so, we propose the following method: at each sample time $t$, a local harmonic approximation $(\tilde{x}_r(t), \tilde{u}_r(t))$ of the reference $(\xr(t), \ur(t))$ is computed to satisfy
\begin{align*}
    &\tilde{x}_r(t) = \xr(t), \, \tilde{x}_r(t{+}N) = \xr(t{+}N), \, \tilde{x}_r'(t{+}N) = \xr'(t{+}N), \\
    &\tilde{u}_r(t) = \ur(t), \, \tilde{u}_r(t{+}N) = \ur(t{+}N), \, \tilde{u}_r'(t{+}N) = \ur'(t{+}N),
\end{align*}
where $a'(t)$ is the time-derivative of $a$ evaluated at time $t$, and is provided as the reference to the HMPC controller.
The objective is to obtain a local reference that provides a good approximation of the desired reference trajectory $(\xr(\cdot), \ur(\cdot))$ between the current sample time $t$ and $t+N$.
We find that the best choice of $w$ is to follow the guidelines from \cite[\S VI]{krupa_harmonic_2022}.

\section{Case study} \label{sec:results}

We show various numerical results of the application of the HMPC formulation to control the ball and plate system described in \cite[\S V.A]{krupa_harmonic_2022}.
The control objective of this system is to control the position of a solid ball that rests on a horizontal plate.
To do so, the inclination of the plate can be manipulated using two independent motors located on each of its main axes.

To improve the numerical conditioning of the solvers, we scale the inputs by a factor of $50$.
We take the HMPC ingredients as $Q = \diag(10, 5, 5, 5, 10, 5, 5, 5)$, $R = 0.5 I_2$, $T_e = 50 Q$, $T_h = 0.1 T_e$, $S_e = 10 I_2$, $S_h = 0.5 S_e$, $N = 8$, and include an additional constraint on the position of the ball on the plate in the form of a regular hexagon with vertices at a distance of $1$ meter from the origin.

We also consider the following MPC formulations:
\begin{itemize}
    \item The \emph{periodic MPC for Tracking} (\emph{perMPCT}) formulation~\cite{limon_mpc_2016}, whose offset cost function matrices we take as matrices $T_e$ and $S_e$ of the HMPC formulation, cost function matrices $Q$ and $R$ as the ones of the HMPC formulation, and prediction horizon also as $N = 8$.
    \item The standard MPC formulation with terminal equality constraint \cite[Eq. (8)]{krupa_implementation_2021}, which we label \emph{equMPC}, but considering a trajectory reference instead of a steady-state reference. We take $N = 16$, and the cost function matrices $Q$ and $R$ as the ones of the HMPC formulation.
        This formulation is the classical MPC for generic reference tracking with a terminal equality constraint and no terminal cost (c.f. \cite[\S 2]{kohler_AUT_2020}); a typical approach when considering generic references, since it avoids the need of computing a terminal invariant set.
        The issue with this formulation is that it does not guarantee recursive feasibility.
        In particular, feasibility is lost if the reference does not satisfy the system constraints.
\end{itemize}

\subsection{Tracking a harmonic reference} \label{sec:results:harmonic}

{\renewcommand{\arraystretch}{1.0}%
    \begin{table*}[t]
    \setlength{\tabcolsep}{2.2pt}
    \centering
	\begin{threeparttable}
    \begin{tabular}{l|lrrrrrrrrrrrrrrrrrrrr}
        \multicolumn{2}{l}{}& \multicolumn{9}{c}{Admissible reference} & \multicolumn{9}{c}{Non-admissible reference} \\
        \cmidrule(lr){3-11}\cmidrule(lr){12-20}
        \multicolumn{2}{l}{}& \multicolumn{4}{c}{Computation time [ms]} & \multicolumn{4}{c}{Number of iterations} & Perf.~\eqref{eq:performance} & \multicolumn{4}{c}{Computation time [ms]} & \multicolumn{4}{c}{Number of iterations} & Perf.~\eqref{eq:performance} \\
        \cmidrule(lr){3-6}\cmidrule(lr){7-10}\cmidrule(lr){11-11}\cmidrule(lr){12-15}\cmidrule(lr){16-19}\cmidrule(lr){20-20}
        \multicolumn{1}{r}{} & MPC (solver) & Avrg. & Med. & Max. & Min. & Avrg. & Med. & Max. & Min. & $\Psi_{20}$ & Avrg. & Med. & Max. & Min. & Avrg. & Med. & Max. & Min. & $\Psi_{20}$\\
        \cmidrule(lr){2-2} \cmidrule(lr){3-20}
        \multirow{4}{*}{\rotatebox[origin=c]{90}{{\scriptsize Harmonic ref.}}}
                             & HMPC (SPCIES) & 0.22 & 0.22 & 0.38 & 0.19 & 28.83 & 29.0 & 29 & 25 & $91.31$
                                             & 4.69 & 4.82 & 5.87 & 0.27 & 636.41 & 653.0 & 733 & 32 & $1739.49$ \\
                             & HMPC (SCS) & 1.64 & 1.52 & 4.17 & 1.38 & 131.87 & 125.0 & 400 & 125 & $91.28$
                                          & 3.63 & 3.55 & 10.80 & 2.87 & 357.61 & 350.0 & 1150 & 275 & $1741.05$ \\
                                          & \emph{perMPCT} (OSQP) & 3.11 & 2.98 & 18.99 & 2.58 & 105.08 & 100.0 & 725 & 100 & $99.86$
                                              & 4.08 & 3.92 & 29.81 & 3.07 & 150.51 & 150.0 & 1200 & 125 & $1733.32$ \\
                                               & \emph{equMPC} (OSQP) & 0.63 & 0.58 & 6.32 & 0.50 & 55.12 & 50.0 & 775 & 50 & $88.86$
                                             & - & - & - & - & - & - & - & - & - \\
        \cmidrule(lr){2-20}
        \multirow{3}{*}{\rotatebox[origin=c]{90}{{\scriptsize Generic ref. }}}
                             & HMPC (SPCIES) & 0.12 & 0.12 & 0.19 & 0.08 & 15.09 & 16.0 & 17 & 11 & $55.30$
                                             & 0.75 & 0.14 & 4.22 & 0.09 & 100.11 & 17.0 & 512 & 13 & $268.40$ \\
                             & HMPC (SCS) & 6.49 & 6.51 & 11.77 & 1.83 & 678.92 & 675.0 & 1250 & 175 & $55.40$
                                          & 8.21 & 7.17 & 24.96 & 1.83 & 897.99 & 775.0 & 2575 & 175 & $269.03$ \\
                                          & \emph{perMPCT} (OSQP) & 7.43 & 6.39 & 36.74 & 3.84 & 148.89 & 125.0 & 700 & 75 & $67.76$
                                              & 22.82 & 22.71 & 31.33 & 14.39 & 532.08 & 525.0 & 750 & 325 & $235.04$ \\
                                                & \emph{equMPC} (OSQP) & 0.61 & 0.57 & 6.48 & 0.48 & 54.74 & 50.0 & 925 & 50 & $48.38$
                                              & - & - & - & - & - & - & - & - & - \\
        \cmidrule(lr){2-20}
    \end{tabular}
\caption{Computational results of the simulations shown in Fig.~\ref{fig:admissible}-\ref{fig:multiple:non}.}
    \label{tab:result}
	\end{threeparttable}
\end{table*}}

We show results comparing the above MPC formulations to track a harmonic reference with base frequency $w = \pi/16$, whose period is therefore of $T = 32$ samples.

We solve \eqref{eq:HMPC} using the solver presented in \cite{krupa_efficiently_2022}, available in the SPCIES toolbox \cite[\texttt{v0.3.11}]{krupa_spcies_2020}, taking its parameter $\rho = 150$.
We also solve \eqref{eq:HMPC} using version \texttt{3.2.3} of the SCS solver \cite{ODonoghue_SCS_21}.
We solve all other MPC formulations using the OSQP solver \texttt{v0.6.2} \cite{stellato_osqp_2020}.
We take the exit tolerances of SPCIES and OSQP as $10^{-4}$, and the ones of SCS as $10^{-6}$, since it is the largest tolerance for which we got reasonable suboptimal solutions.
Tests are performed on a $2.3$GHz Intel~i7 in MATLAB using the C-MEX interface of the solvers.

Fig.~\ref{fig:admissible} shows the closed-loop trajectory of the system when tracking an admissible harmonic reference (depicted in red).
We show the results when using the HMPC and \emph{perMPCT} formulations in Fig.~\ref{fig:admissible:HMPC} and~\ref{fig:admissible:perMPCT}, respectively, as well as the control inputs for both formulations in Fig.~\ref{fig:admissible:input}.
The results indicate that the HMPC controller behaves similarly to the \emph{perMPCT} controller when the reference is admissible.

Fig.~\ref{fig:non-admissible} is analogous to Fig.~\ref{fig:admissible} but taking a non-admissible harmonic reference.
Fig.~\ref{fig:non-admissible:Te} shows the result using \emph{perMPCT} and the HMPC described above, where $T_h = 0.1 T_e$.
Fig.~\ref{fig:non-admissible:Th} shows the results with HMPC if we instead take $T_h = 100 T_e$.
Finally, Fig.~\ref{fig:non-admissible:input} shows the control inputs corresponding to Fig.~\ref{fig:non-admissible:Te}.
Comparing Fig.~\ref{fig:non-admissible:Te} and \ref{fig:non-admissible:Th} we see that when the $T_e$ and $S_e$ terms are dominant, HMPC behaves similarly to \emph{perMPCT}.
However, when $T_h$ and $S_h$ are dominant, the behavior of the closed-loop system changes drastically, as shown in Fig.~\ref{fig:non-admissible:Th}.
The reason behind this behavior is the offset cost function $V_h$, which does not penalize the discrepancy between the artificial reference with the desired reference, but instead between the parameters that characterize them, as discussed in Remark~\ref{rem:shape}.
This behavior differs from the one expected from classical MPC formulations. 
However, it might be very interesting for those applications in which maintaining the shape of the trajectory is more important than being closer to the reference at each individual sample time.
Some potential applications are power electronics, where we want the output to be as close as possible to a perfect harmonic signal, possibly at the cost of decreasing the power-output, or aerospace rendezvous.

\subsection{Tracking an arbitrary periodic reference} \label{sec:results:arbitrary}

\begin{figure*}[hbtp]
    \centering
    \begin{subfigure}[ht]{0.26\textwidth}
        \includegraphics[width=\linewidth]{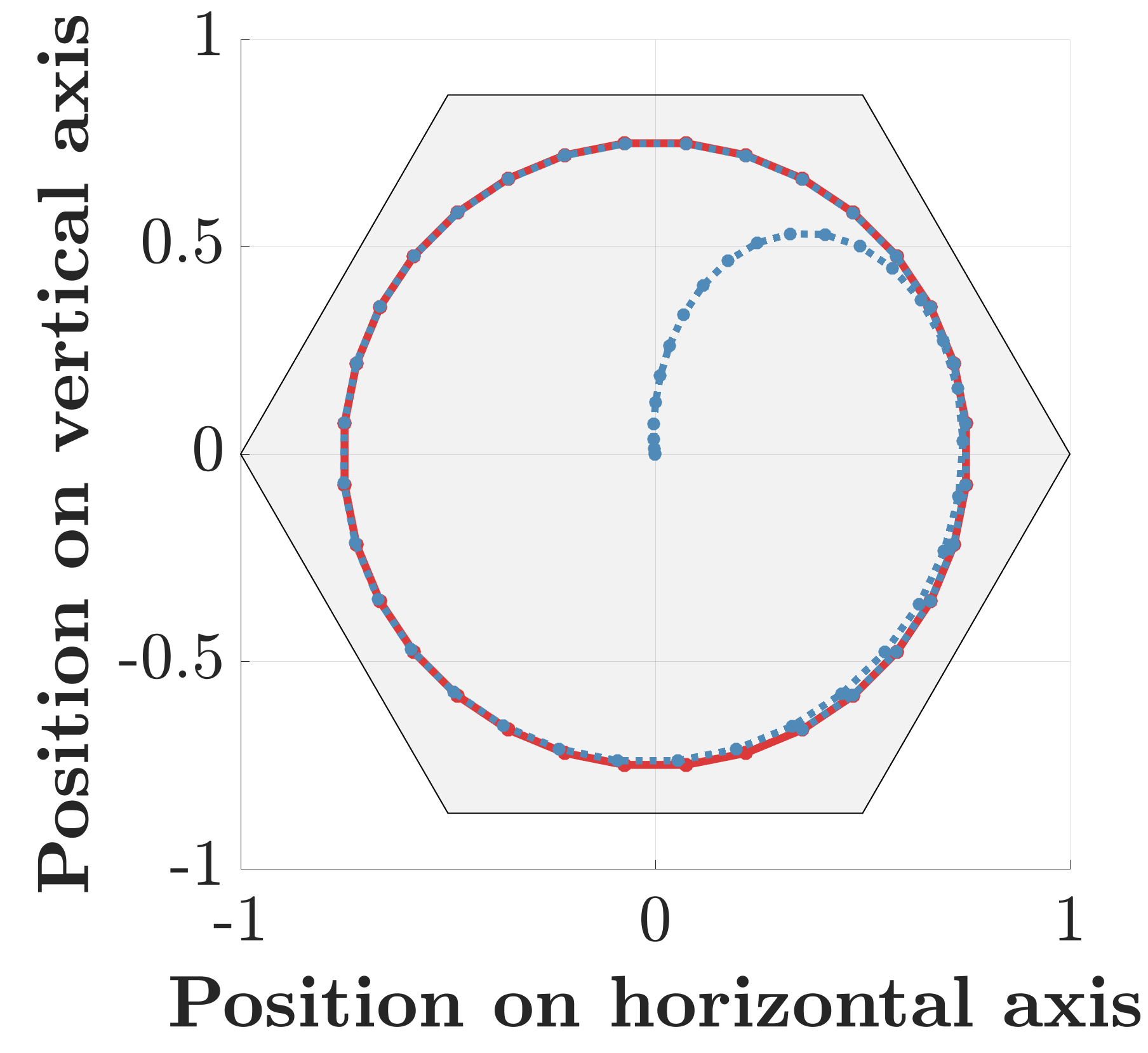}
        \caption{Position using HMPC.}
        \label{fig:admissible:HMPC}
    \end{subfigure}%
    \hfill
    \begin{subfigure}[ht]{0.26\textwidth}
        \includegraphics[width=\linewidth]{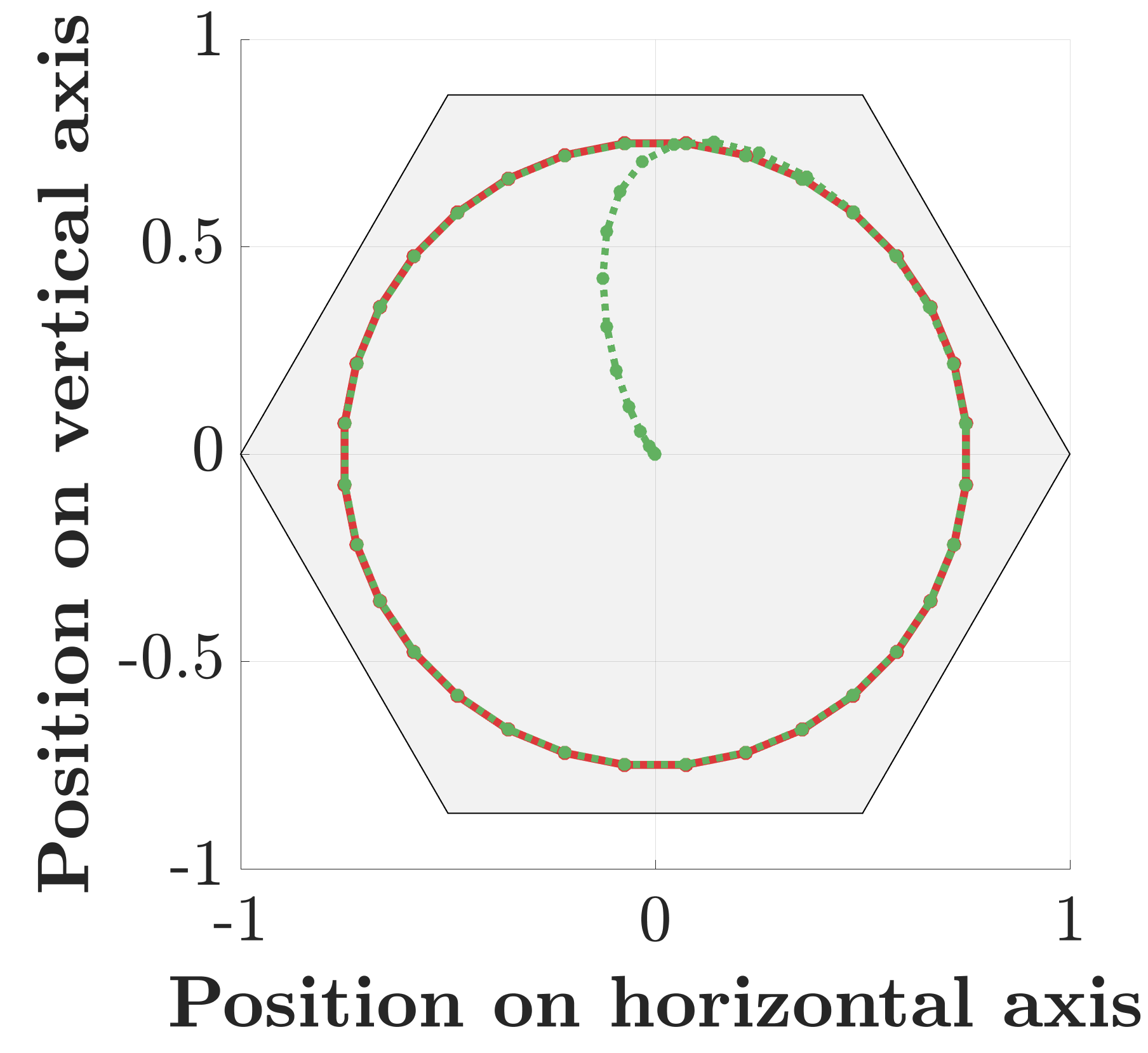}
        \caption{Position using \emph{perMPCT}.}
        \label{fig:admissible:perMPCT}
    \end{subfigure}%
    \hfill
    \begin{subfigure}[ht]{0.43\textwidth}
        \includegraphics[width=1\linewidth]{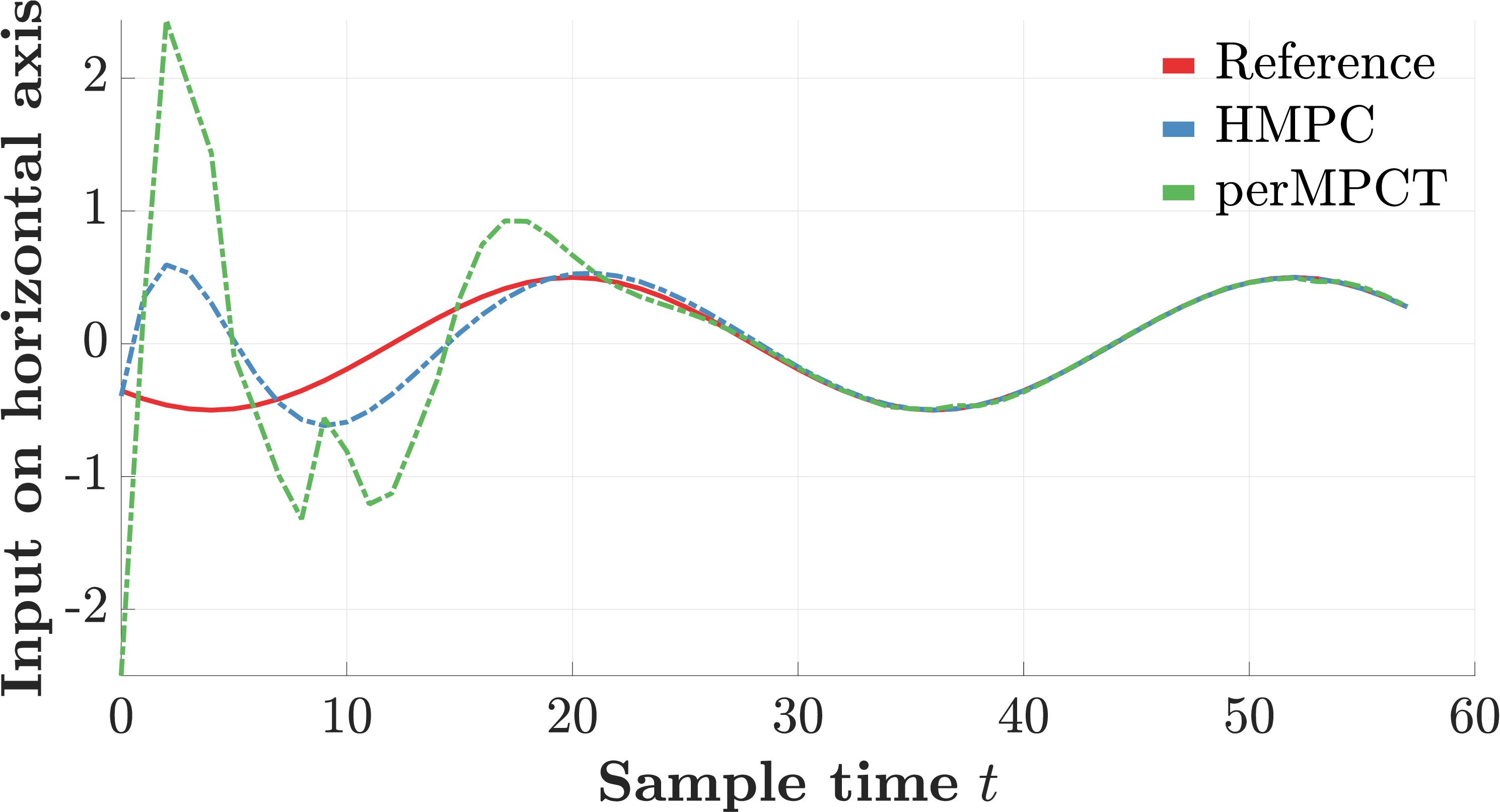}
        \caption{Input trajectories.}
        \label{fig:admissible:input}
    \end{subfigure}%
    \hfill
    \caption{Tracking an admissible harmonic reference.}
    \label{fig:admissible}
\end{figure*}

\begin{figure*}[hbtp]
    \centering
    \begin{subfigure}[ht]{0.26\textwidth}
        \includegraphics[width=\linewidth]{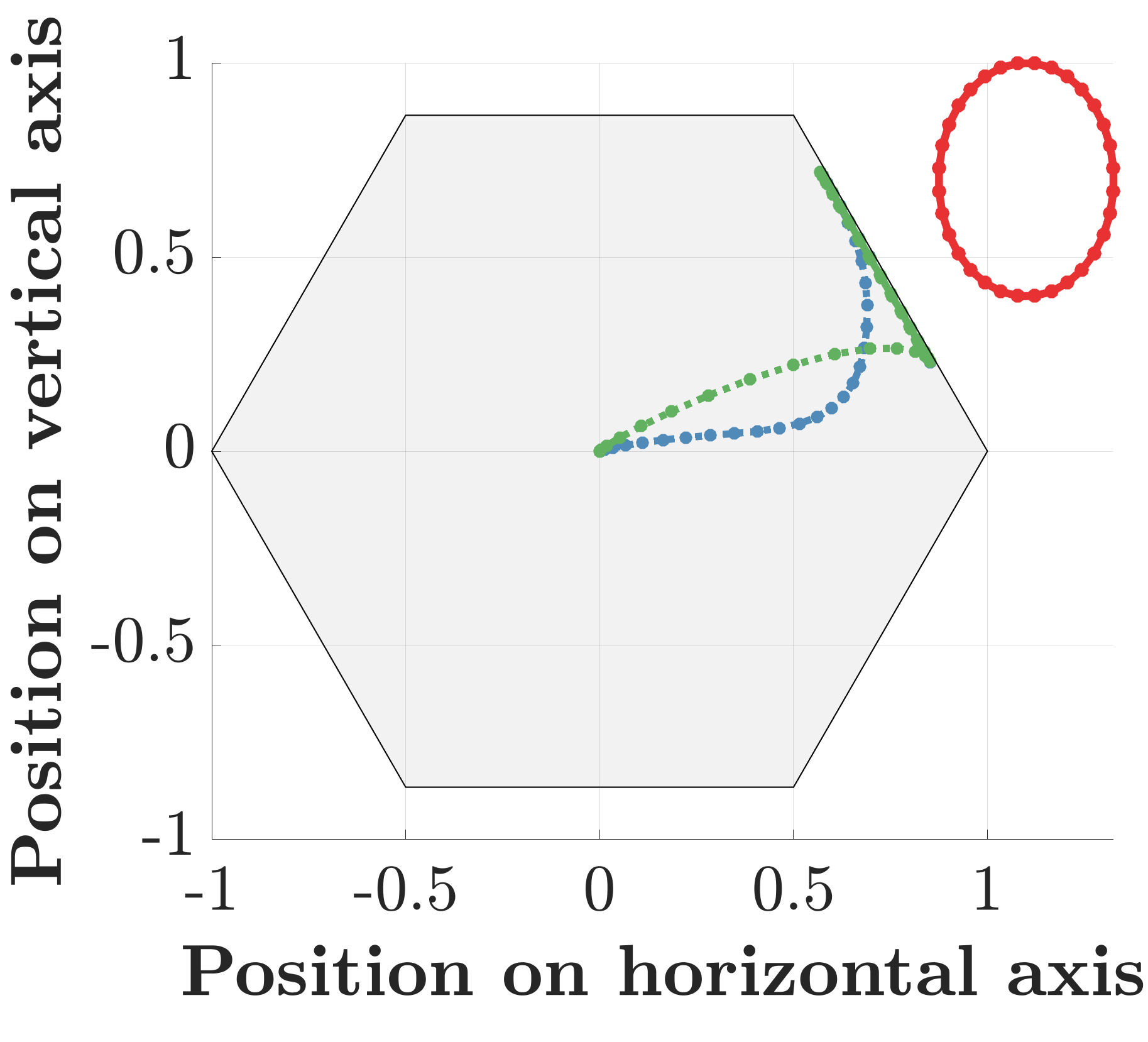}
        \caption{Position using $T_h = 0.1 T_e$.}
        \label{fig:non-admissible:Te}
    \end{subfigure}%
    \hfill
    \begin{subfigure}[ht]{0.26\textwidth}
        \includegraphics[width=1\linewidth]{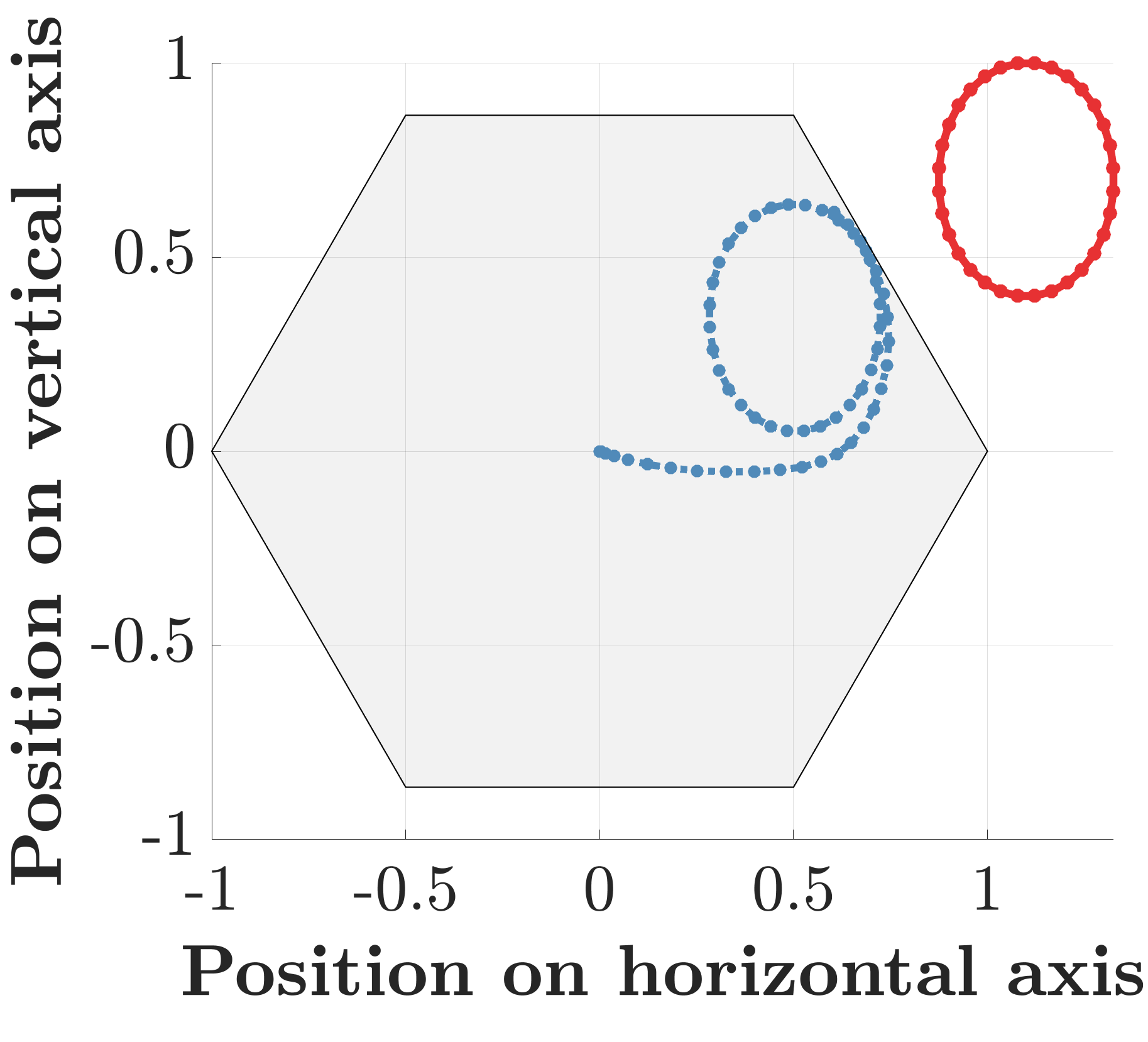}
        \caption{Position using $T_h = 100 T_e$.}
        \label{fig:non-admissible:Th}
    \end{subfigure}%
    \hfill
    \begin{subfigure}[ht]{0.43\textwidth}
        \includegraphics[width=1\linewidth]{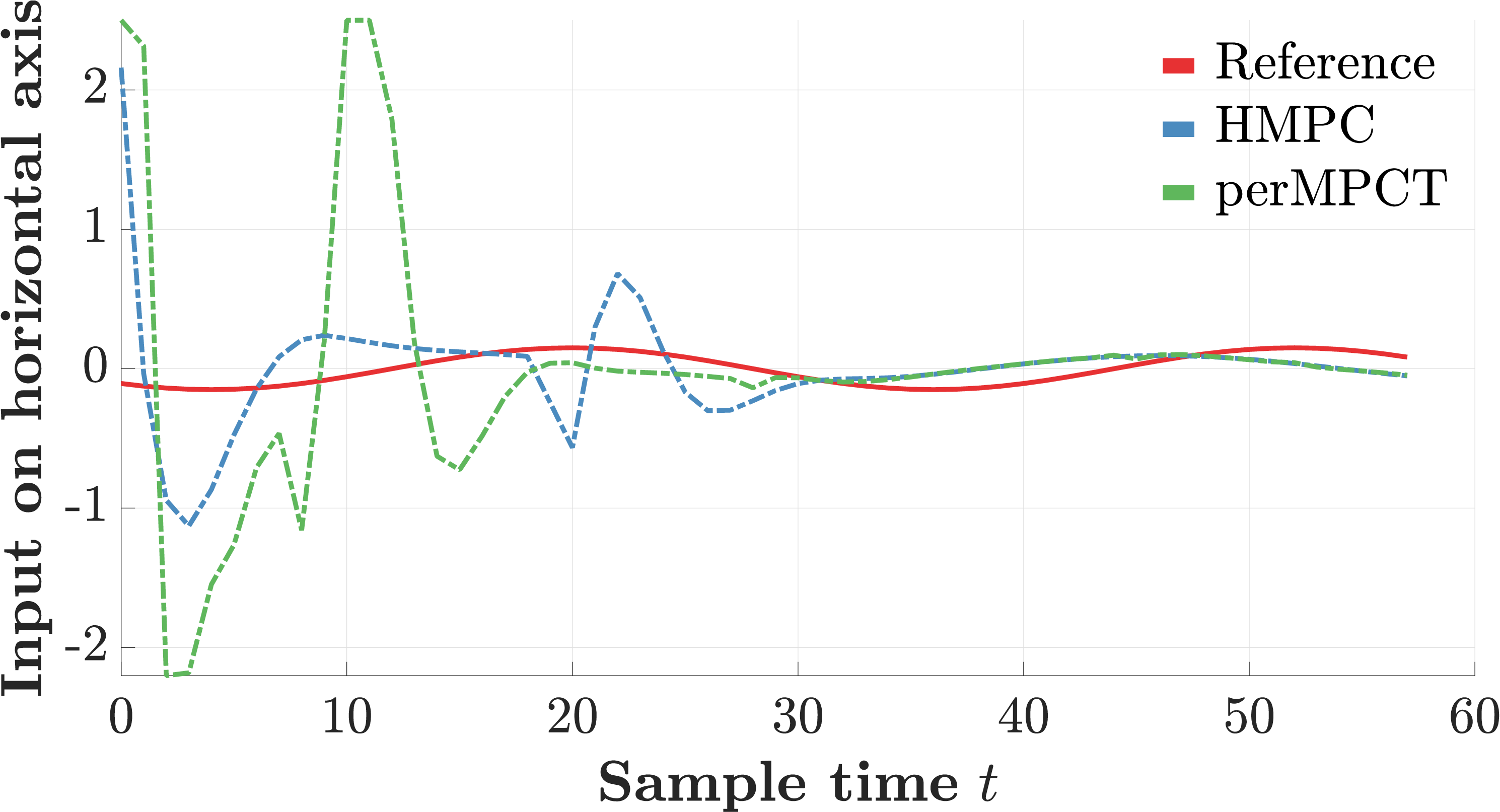}
        \caption{Input trajectories (for $T_h = 0.1 T_e$).}
        \label{fig:non-admissible:input}
    \end{subfigure}%
    \hfill
    \caption{Tracking a non-admissible harmonic reference.}
    \label{fig:non-admissible}
\end{figure*}

\begin{figure*}[hbtp]
    \centering
    \begin{subfigure}[ht]{0.23\textwidth}
        \includegraphics[width=\linewidth]{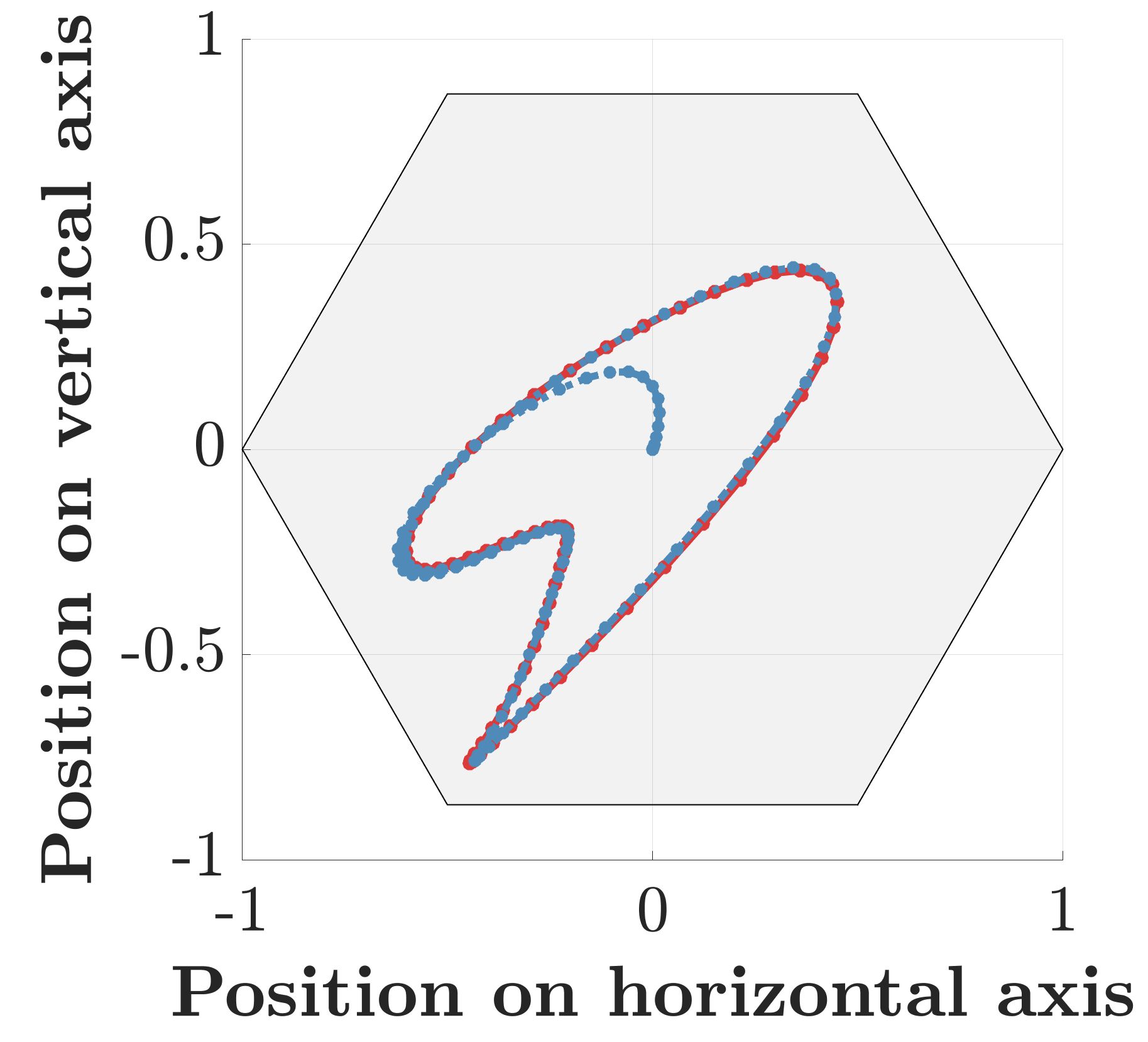}
        \caption{Ball position using HMPC.}
        \label{fig:multiple:admissible:position}
    \end{subfigure}%
    \hfill
    \begin{subfigure}[ht]{0.38\textwidth}
        \includegraphics[width=1\linewidth]{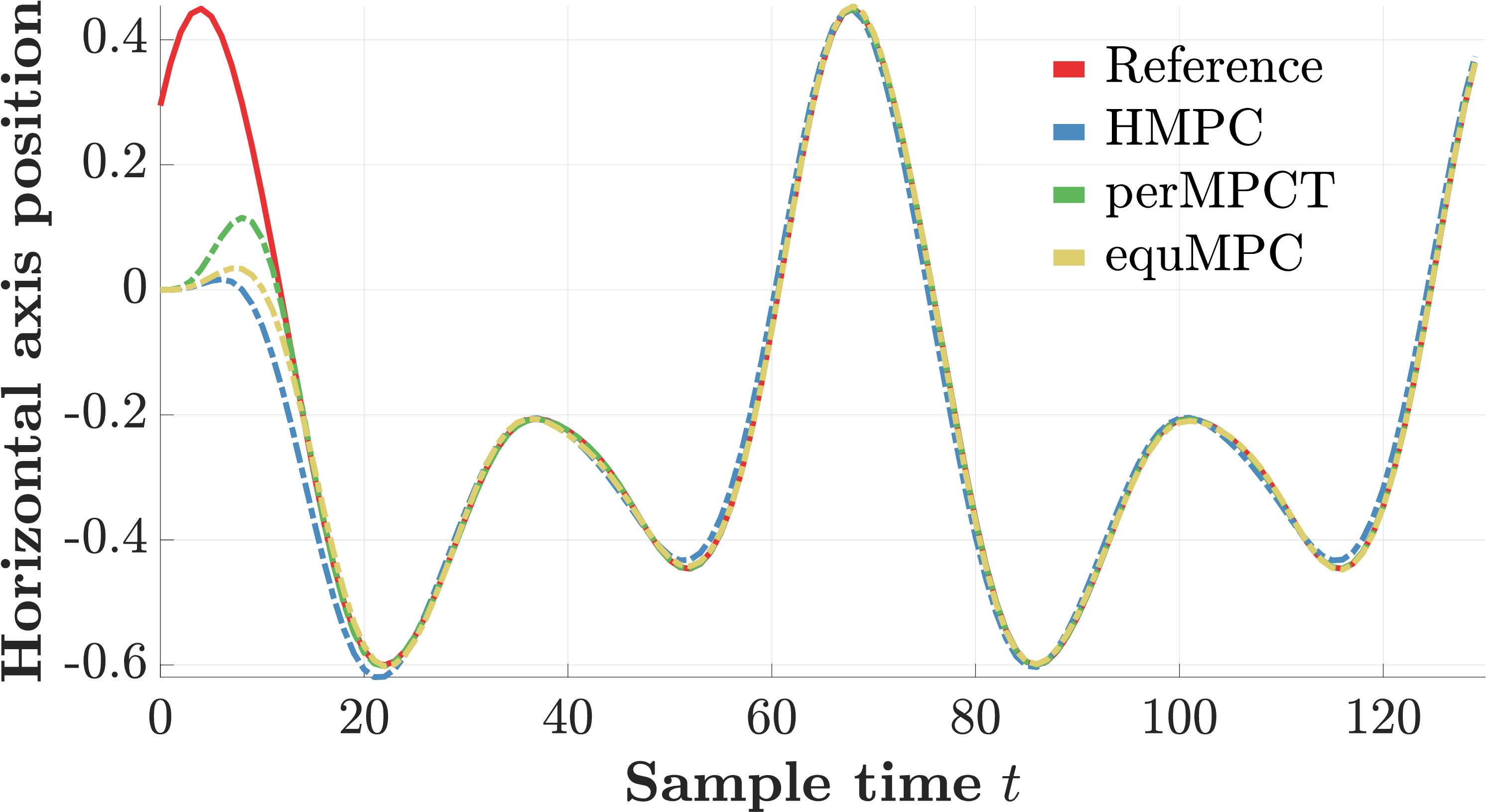}
        \caption{State trajectories.}
        \label{fig:multiple:admissible:state}
    \end{subfigure}%
    \hfill
    \begin{subfigure}[ht]{0.38\textwidth}
        \includegraphics[width=1\linewidth]{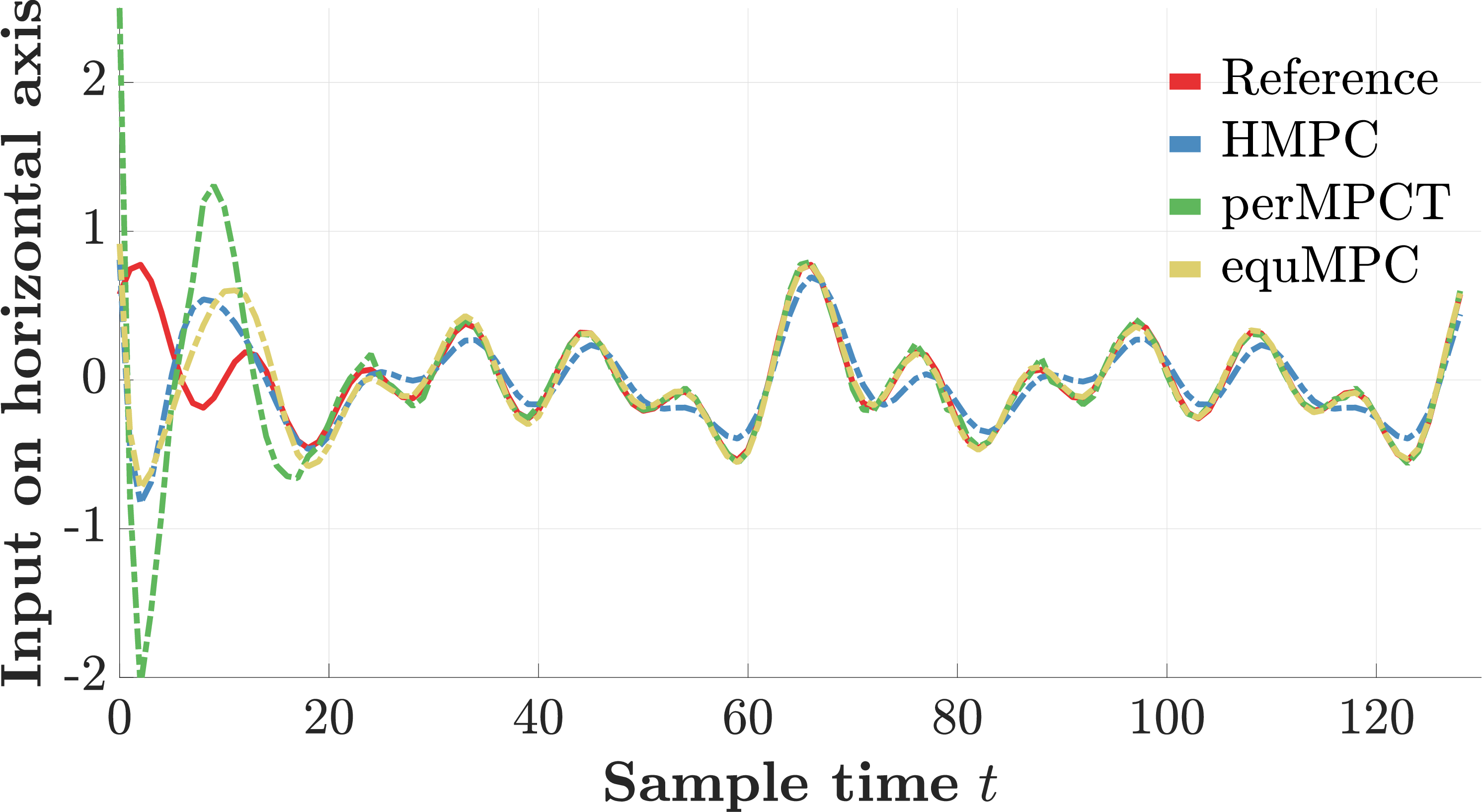}
        \caption{Input trajectories.}
        \label{fig:multiple:admissible:input}
    \end{subfigure}%
    \hfill
    \caption{Tracking an admissible arbitrary reference.}
    \label{fig:multiple:admissible}
\end{figure*}

\begin{figure*}[hbtp]
    \centering
    \begin{subfigure}[ht]{0.23\textwidth}
        \includegraphics[width=\linewidth]{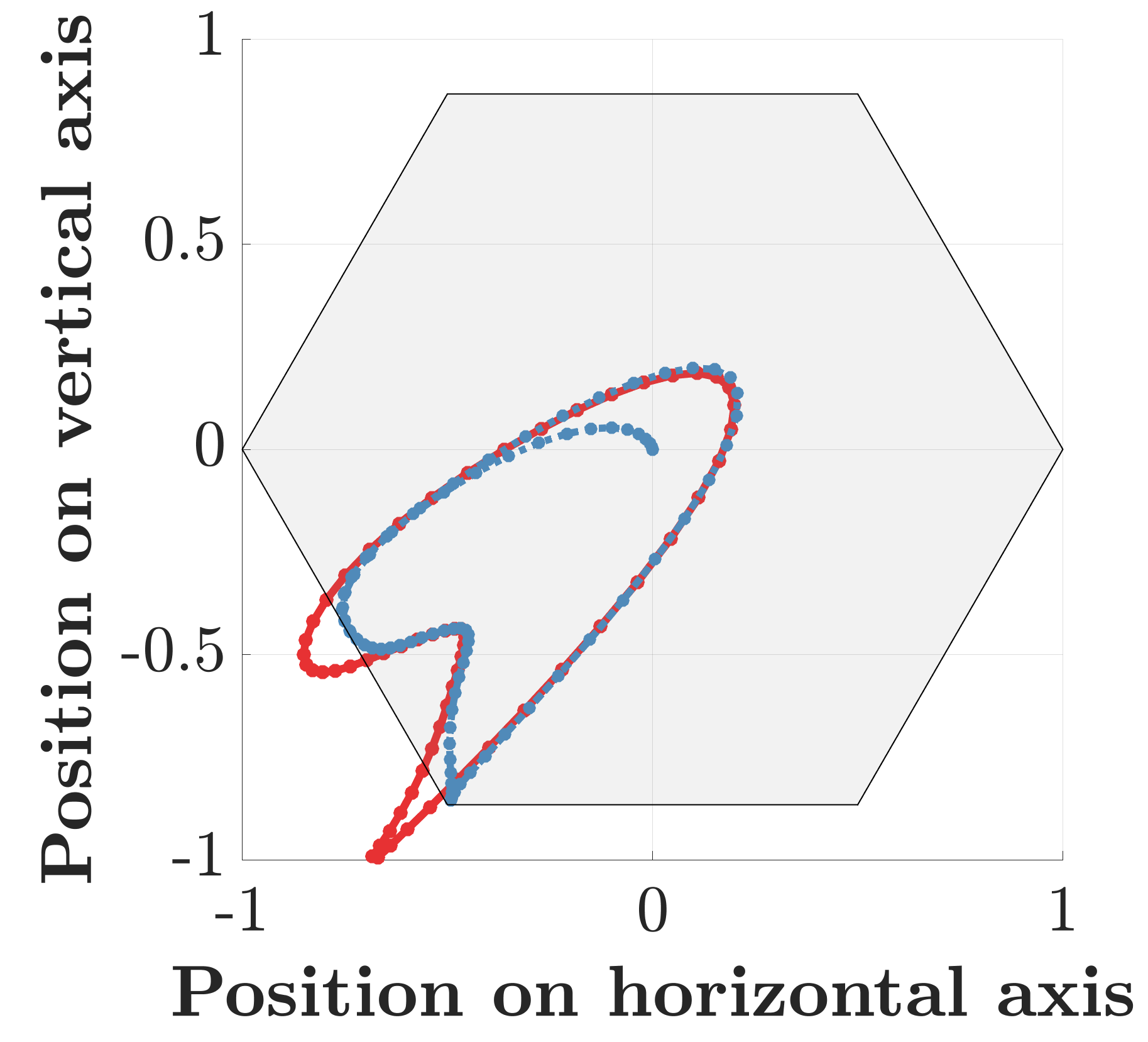}
        \caption{Ball position using HMPC.}
        \label{fig:multiple:non:position}
    \end{subfigure}%
    \hfill
    \begin{subfigure}[ht]{0.38\textwidth}
        \includegraphics[width=1\linewidth]{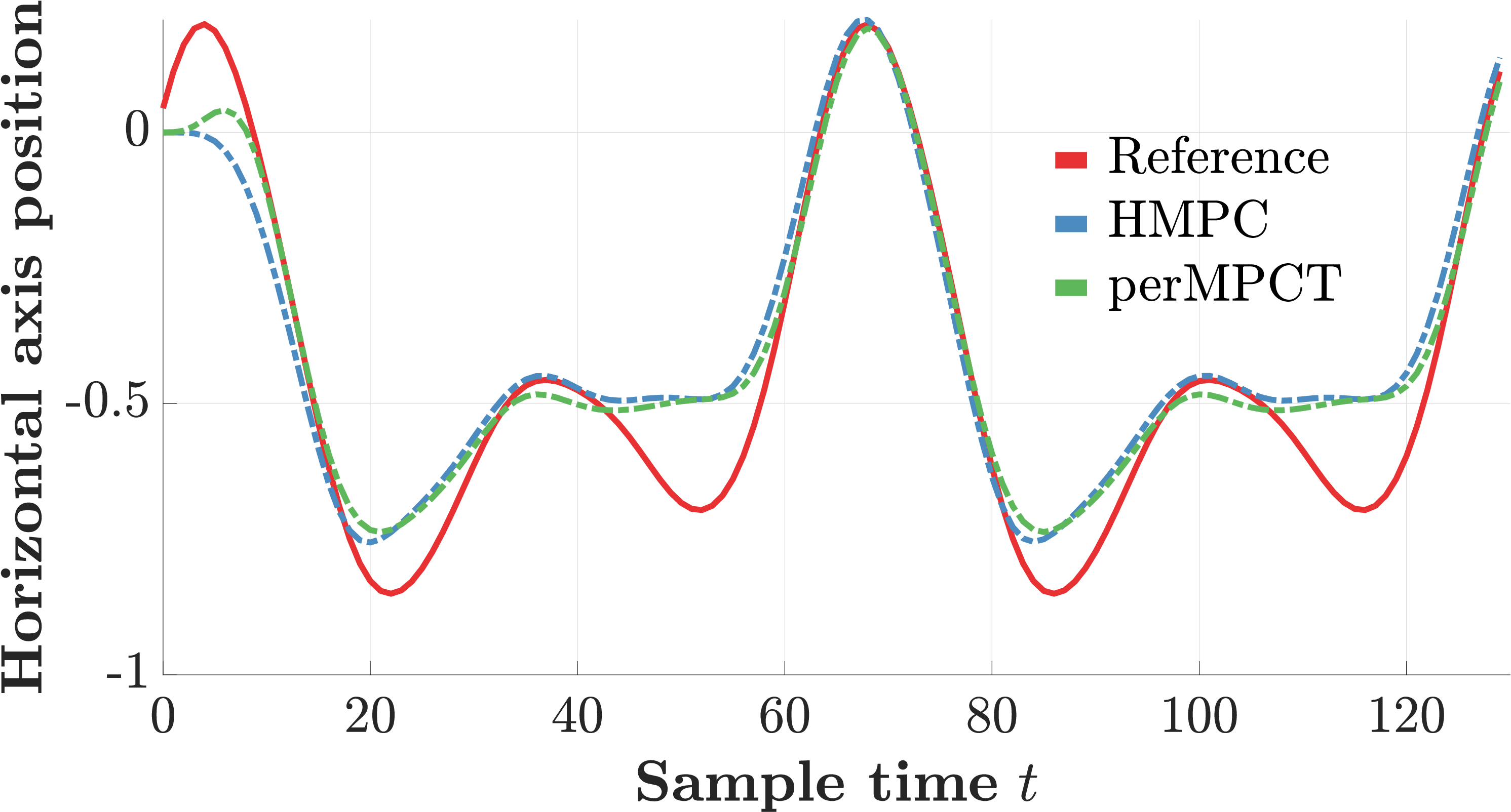}
        \caption{State trajectories.}
        \label{fig:multiple:non:state}
    \end{subfigure}%
    \hfill
    \begin{subfigure}[ht]{0.38\textwidth}
        \includegraphics[width=1\linewidth]{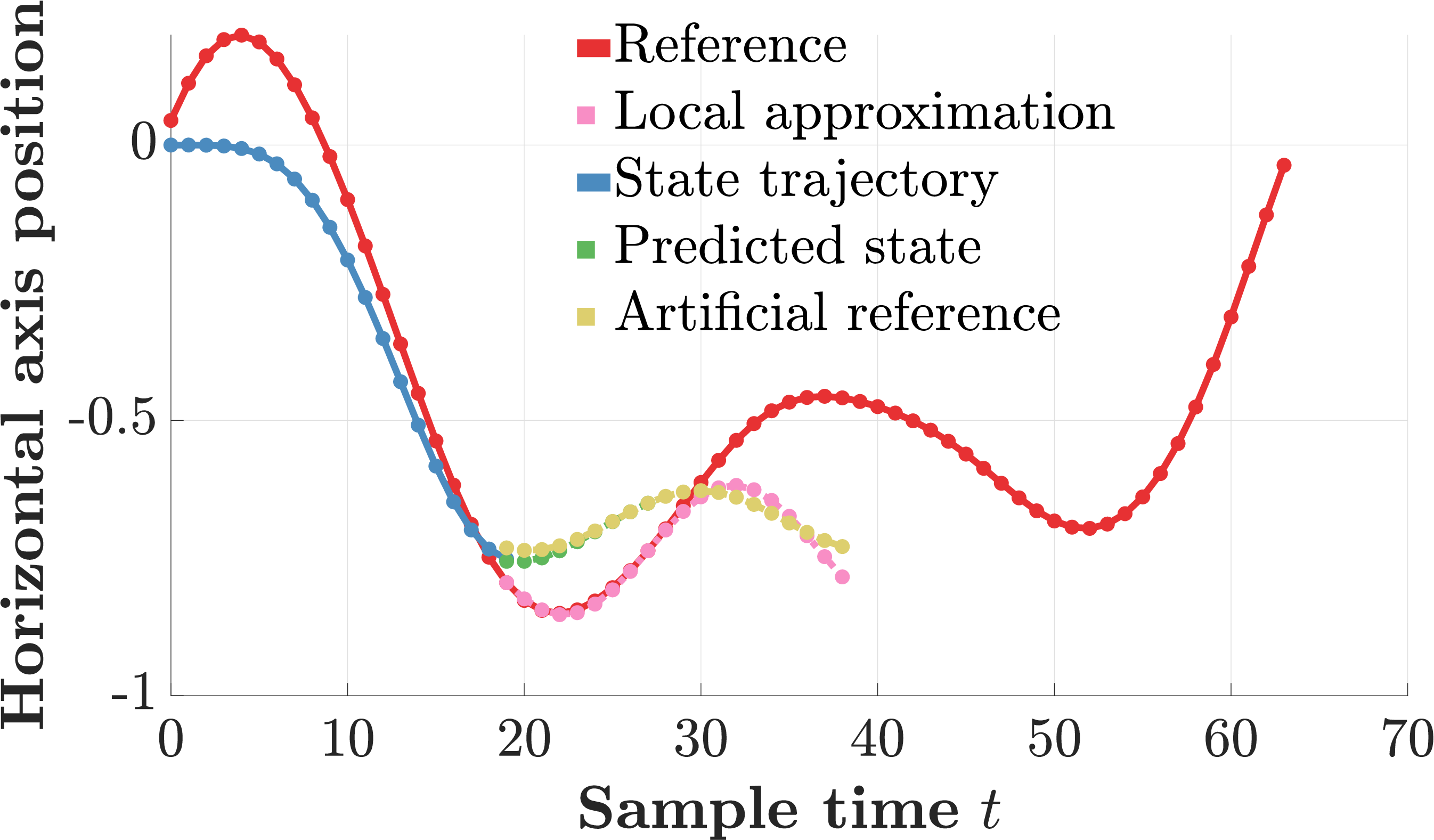}
        \caption{Snapshot of HMPC at $t = 20$.}
        \label{fig:multiple:non:axis1}
    \end{subfigure}%
    \hfill
    \caption{Tracking a non-admissible arbitrary reference.}
    \label{fig:multiple:non}
\end{figure*}

We now show the results when controlling a reference which is not given by a harmonic signal.
Instead, we take a multiple harmonic signal on the form
\begin{align*}
    x_r(t) &= x_{re} + \Sum{i = 1}{p} x_{rs,i} \sin(i w_r t) + x_{rc,i} \cos( i w_r t), \\
    u_r(t) &= u_{re} + \Sum{i = 1}{p} u_{rs,i} \sin(i w_r t) + u_{rc,i} \cos( i w_r t)
\end{align*}
satisfying $\xr(t+1) = A \xr(t) + B \ur(t)$, $\forall t$, where we select the base frequency of the reference as $w_r = \pi/32$ ($T = 64$) and the number of harmonics as $p = 6$.
The main control objective is to control the position of the ball on the plate.
To this end, we focus on penalizing error in position tracking by taking $Q = \diag(10, 0.5, 0.5, 0.5, 10, 0.5, 0.5, 0.5)$.
We also find that $T_h$ needs to be larger when tracking arbitrary references, so we take $T_h = T_e$.
Finally, following the guidelines of \cite[\S VI]{krupa_harmonic_2022}, we take the base frequency as $w = 0.3254$.

Fig.~\ref{fig:multiple:admissible:position} shows the trajectory of the ball on the plate using HMPC.
Fig.~\ref{fig:multiple:admissible:state} and~\ref{fig:multiple:admissible:input} show the position and input trajectories of the horizontal axis, respectively, using the MPC formulations described at the beginning of this section.
Fig.~\ref{fig:multiple:non:position} and~\ref{fig:multiple:non:state} show analogous results to Fig.~\ref{fig:multiple:admissible:position} and~\ref{fig:multiple:admissible:state} but shifting the reference so that it is (partially) non-admissible.
We note that \emph{equMPC} is not included because it looses feasibility when the reference does not satisfy the system constraints.
Finally, Fig.~\ref{fig:multiple:non:axis1} shows a snapshot, at sample time $t = 20$, of the position on the horizontal axis corresponding to the test shown in Fig.~\ref{fig:multiple:non:position}.
The figure includes the local harmonic approximation and the predicted states and artificial harmonic reference returned by the HMPC solver.
Note that the local harmonic reference approximates the desired reference.

The results show that the HMPC formulation tracks the admissible generic reference reasonably well.
When the reference is non-admissible, the local harmonic reference approximation and the artificial reference are no longer close to each other in the areas in which the reference is non-admissible.
In this case, the closed-loop trajectory resembles the desired reference in the areas in which it is non-admissible, although in the case of generic references this resemblance is no longer guaranteed, even if $T_h$ and $S_h$ are dominant.
One of the potential applications of this paradigm is to track generic reference trajectories that are not completely known before hand, since the proposed approach only requires knowledge of the future $N$ elements of the reference.
The advantage of HMPC in this case is its guaranteed recursive feasibility.

\subsection{Computational results and performance} \label{sec:results:computation}

Table~\ref{tab:result} shows the computational results obtained with the different MPC formulations and solvers used in the tests shown in Fig.~\ref{fig:admissible}-\ref{fig:multiple:non}.
Results for Fig.~\ref{fig:admissible}-\ref{fig:non-admissible} are in the rows labeled with ``Harmonic ref.", and the ones for Fig.~\ref{fig:multiple:admissible}-\ref{fig:multiple:non} in the rows labeled with ``Generic ref.".
The computation times for HMPC related to Fig.~\ref{fig:non-admissible} consider $T_h = 0.1 T_e$ (as in Fig.~\ref{fig:non-admissible:Te}).
Table~\ref{tab:result} also shows the performance of the formulations, measured as
\begin{equation} \label{eq:performance}
\Psi_s = \Sum{t=0}{s T-1} \|x(t) - \xr(t) \|_Q^2 + \|u(t) - \ur(t)\|_R^2,
\end{equation}
where we recall that $T$ is the period of the reference.

The performance of HMPC when tracking a harmonic reference is better than with \emph{perMPCT} when using $N = 8$.
The higher performance of the HMPC formulation when working with small prediction horizons was reported in \cite{krupa_harmonic_2022} for the case of tracking constant references.
The result indicates that this may also be the case when tracking harmonic references.

The HMPC and \emph{perMPCT} formulations do not perform as well as \emph{equMPC} when the reference is admissible.
However, the advantage of the HMPC and \emph{perMPCT} formulations is that they have guaranteed recursive feasibility and a larger domain of attraction.
Indeed, we note that the prediction horizon of \emph{equMPC} was set to $N = 16$ because any smaller value resulted in a loss of feasibility in the test shown in Fig.~\ref{fig:multiple:admissible}.
The performance results for the tests with non-admissible reference highlight the good performance of the HMPC formulation when compared to \emph{perMPCT}.
The computational results also highlight the good performance of the tailored HMPC solver available in~\cite{krupa_spcies_2020}.
In particular, the fact that the complexity of the optimization problem does not depend on the period of the reference can provide significant computational benefits.

\section{Conclusions} \label{sec:conclusions}

This article has presented an extension of the HMPC formulation \cite{krupa_harmonic_2022} for tracking harmonic references and has discussed its application for tracking arbitrary reference trajectories.
In the case of tracking harmonic references, we showed that the terminal ingredients can be chosen to penalize deviations with respect to the ``shape'' of the reference, which is an interesting property that may have useful practical applications.
Preliminary simulations indicate that the HMPC formulation can provide good tracking of arbitrary references if its ingredients are chosen appropriately.
An interesting application of this paradigm is when only the future $N$ reference values are known at any given instant.
Computational results indicate that the modified HMPC solver from \cite{krupa_efficiently_2022} provides computational times that are competitive with state-of-the-art solvers.
Moreover, the computation time per iteration of the solver does not depend on the period of the reference, making it an ideal candidate when working with references with large periods.


\begin{appendix}

\begin{proof}[Proof of Lemma~\ref{lemma:xH}]
At $t+1$, $(\xHo(t+1), \uHo(t+1))$ is the optimal solution of 
\begin{subequations} \label{eq:proof:xH:OP:original}
\begin{align}
    \min\limits_{\xH, \uH}\;& V_h(\xH, \uH, \Twa{n_x} \xrH(t), \Twa{n_u} \urH(t)) \\
    {\rm s.t.} \;& (\xH, \uH) \in \cD \cap \cC,
\end{align}
\end{subequations}
where we recall that $\Twa{m} \doteq \diag(I_{m}, \Tw{m})$ and $\Tw{m}$ is defined in \eqref{eq:def:Tw}.
By taking the change of variables $\xH = \Twa{n_x} \hat{\vv{x}}_h$ and ${\uH = \Twa{n_u} \hat{\vv{u}}_h}$, we can recover the optimal solution of \eqref{eq:proof:xH:OP:original} from the optimal solution of
\begin{equation}
\label{eq:proof:xH:OP:transformed}
\min\limits_{\hat{\vv{x}}_h, \hat{\vv{u}}_h}\; \left\{ \tilde{V}_h,\; 
{\rm s.t.} \, (\Twa{n_x} \hat{\vv{x}}_h, \Twa{n_u} \hat{\vv{u}}_h) \in \cD \cap \cC \right\},
\end{equation}
where $\tilde{V}_h \doteq V_h(\Twa{n_x} \hat{\vv{x}}_h, \Twa{n_u} \hat{\vv{u}}_h, \Twa{n_x} \xrH(t), \Twa{n_u} \urH(t))$.

From Propositions~\ref{prop:harmonic:dynamics} and \ref{prop:harmonic:constraints}, we have that the set $\cD \cap \cC$ is closed under the transform $(\Twa{n_x}, \Twa{n_u})$, i.e., 
\begin{equation*}
(\Twa{n_x} \hat{x}_h, \Twa{n_u} \hat{u}_h) \in \cD \cap \cC \iff (\hat{x}_h, \hat{u}_h) \in \cD \cap \cC.
\end{equation*}

Additionally, the cost function satisfies the equality ${\tilde{V}_h = V_h(\hat{\vv{x}}_h, \hat{\vv{u}}_h, \xrH(t), \urH(t))}$.
Indeed, first note that the $\xe$ and $\ue$ terms of $V_h$ are equal due to the identity matrices in $\Twa{n_x}$ and $\Twa{n_u}$.
Next, for the $\xs$, $\xc$, $\us$, and $\uc$ terms, we have that for any $z, v \in \R^{2 m}$ and $D = \diag(\tilde{D}, \tilde{D}) $, with  $\tilde{D} \in \Dp{m}$,
\begin{align*}
    \| \Tw{m} (z - v) \|^2_D &= (z - v)\T (\Tw{m})\T D \Tw{m} (z - v) \\
                                     &\becauseof{(*)} (z - v)\T D (z - v) = \| z - v \|^2_D,
\end{align*}
where $(*)$ follows from the definition of $\Tw{m}$ \eqref{eq:def:Tw} and the well known identity $\sin^2(w) + \cos^2(w) = 1$.
The equality then follows from the fact that $T_h \in \Dp{n_x}$ and $S_h \in \Dp{n_u}$.
Therefore, \eqref{eq:proof:xH:OP:transformed} is equivalent to \eqref{eq:xHo:OP}, whose solution is $(\xHo(t), \uHo(t))$ by definition, thus $\hat{\vv{x}}_h^* = \xHo(t)$ and $\hat{\vv{u}}_h^* = \uHo(t)$. \qedhere
\end{proof}

\begin{proof}[Proof of Theorem~\ref{theo:feasibility}]
The proof is nearly identical to the recursive feasibility proof of the original HMPC formulation \cite[Theorem~1]{krupa_harmonic_2022} since the constraints of \eqref{eq:HMPC} are identical to the ones of the HMPC formulation presented in \cite{krupa_harmonic_2022} with the exception of constraint \eqref{eq:HMPC:xN}, which instead reads as $x^N = \xe + \xc$ in~\cite{krupa_harmonic_2022}.
This difference, however, is simply a time-shift which we take to simplify the notation when working with harmonic references.
Thus, since the recursive feasibility does not depend on the reference, the proof of the theorem follows identically to the proof of \cite[Theorem~1]{krupa_harmonic_2022} but taking \cite[Eq.~(19b)]{krupa_harmonic_2022} as $\bar{u}^+_{N-1} = \ue + \us \sin(w N) + \uc \cos(w N)$.
\end{proof}

\begin{proof}[Proof of Theorem~\ref{theo:stability}]
The proof is very similar to the asymptotic stability proof of the original HMPC formulation \cite[Theorem~3]{krupa_harmonic_2022}.
As in \cite{krupa_harmonic_2022}, the proof is based on finding a Lyapunov function that satisfies the asymptotic stability conditions from \cite[Theorem~2]{krupa_harmonic_2022}.
The difference is that the Lyapunov function is taken for $x(t) - \xho(t)$, whereas in \cite{krupa_harmonic_2022} the Lyapunov function is taken for $x(t) - \xeo$, since the artificial reference in \cite{krupa_harmonic_2022} is a steady-state.
This requires two modifications to the proof.
The first is that \cite[Lemma~2]{krupa_harmonic_2022} has to be rewritten for the case of a harmonic reference instead of a steady-state reference.
The lemma can be rewritten, with small modifications, to prove that $x(t) = \xe^* + \xc^*$ if and only if $x(t) = \xho(t)$.
The second are the modifications due to the time-varying nature of the reference $(\xr(\cdot), \ur(\cdot))$ and to the difference between the offset cost function $V_h$ of \eqref{eq:HMPC} with the one used in \cite{krupa_harmonic_2022}.
These differences require several modifications which are solved using simple algebraic manipulations, resulting in the same lines of reasoning and arguments used in \cite[Theorem~3]{krupa_harmonic_2022}.
\end{proof}

\end{appendix}

\vspace*{-1em}

\bibliographystyle{IEEEtran}
\bibliography{IEEEabrv,ellipHMPC}

\begin{thebibliography}{10}
\providecommand{\url}[1]{#1}
\csname url@samestyle\endcsname
\providecommand{\newblock}{\relax}
\providecommand{\bibinfo}[2]{#2}
\providecommand{\BIBentrySTDinterwordspacing}{\spaceskip=0pt\relax}
\providecommand{\BIBentryALTinterwordstretchfactor}{4}
\providecommand{\BIBentryALTinterwordspacing}{\spaceskip=\fontdimen2\font plus
\BIBentryALTinterwordstretchfactor\fontdimen3\font minus
  \fontdimen4\font\relax}
\providecommand{\BIBforeignlanguage}[2]{{%
\expandafter\ifx\csname l@#1\endcsname\relax
\typeout{** WARNING: IEEEtran.bst: No hyphenation pattern has been}%
\typeout{** loaded for the language `#1'. Using the pattern for}%
\typeout{** the default language instead.}%
\else
\language=\csname l@#1\endcsname
\fi
#2}}
\providecommand{\BIBdecl}{\relax}
\BIBdecl

\bibitem{rawlings_model_2017}
J.~B. Rawlings, D.~Q. Mayne, and M.~Diehl, \emph{Model predictive control:
  theory, computation, and design}, 2nd~ed.\hskip 1em plus 0.5em minus
  0.4em\relax Madison, Wisconsin: Nob Hill Publishing, 2017.

\bibitem{gupta_period-robust_2006}
M.~Gupta and J.~H. Lee, ``Period-robust repetitive model predictive control,''
  \emph{Journal of Process Control}, vol.~16, no.~6, p. 545–555, 2006.

\bibitem{leomanni_sum--norms_2020}
M.~Leomanni, G.~Bianchini, A.~Garulli, and R.~Quartullo, ``Sum-of-norms {MPC}
  for linear periodic systems with application to spacecraft rendezvous,'' in
  \emph{2020 59th {IEEE} {Conference} on {Decision} and {Control}
  ({CDC})}.\hskip 1em plus 0.5em minus 0.4em\relax IEEE, 2020, p. 4665–4670.

\bibitem{gondhalekar_mpc_2011}
R.~Gondhalekar and C.~N. Jones, ``{MPC} of constrained discrete-time linear
  periodic systems — {A} framework for asynchronous control: {Strong}
  feasibility, stability and optimality via periodic invariance,''
  \emph{Automatica}, vol.~47, no.~2, p. 326–333, 2011.

\bibitem{risbeck_economic_2020}
M.~J. Risbeck and J.~B. Rawlings, ``Economic model predictive control for
  time-varying cost and peak demand charge optimization,'' \emph{IEEE
  Transactions on Automatic Control}, vol.~65, no.~7, p. 2957–2968, 2020.

\bibitem{limon_mpc_2016}
D.~Limon, M.~Pereira, D.~M. de~la Pe{\~n}a, T.~Alamo, C.~N. Jones, and M.~N.
  Zeilinger, ``{MPC} for tracking periodic references,'' \emph{IEEE
  Transactions on Automatic Control}, vol.~61, no.~4, pp. 1123--1128, 2016.

\bibitem{limon_mpc_2008}
D.~Limon, I.~Alvarado, T.~Alamo, and E.~Camacho, ``{MPC} for tracking piecewise
  constant references for constrained linear systems,'' \emph{Automatica},
  vol.~44, no.~9, p. 2382–2387, 2008.

\bibitem{ferramosca_mpc_2009}
A.~Ferramosca, D.~Limon, I.~Alvarado, T.~Alamo, and E.~Camacho, ``{MPC} for
  tracking with optimal closed-loop performance,'' \emph{Automatica}, vol.~45,
  no.~8, p. 1975–1978, 2009.

\bibitem{Kohler_NMPC_18}
J.~K{\"o}hler, M.~A. M{\"u}ller, and F.~Allg{\"o}wer, ``{MPC} for nonlinear
  periodic tracking using reference generic offline computations,''
  \emph{IFAC-PapersOnLine}, vol.~51, no.~20, pp. 556--561, 2018.

\bibitem{kohler_AUT_2020}
------, ``A nonlinear tracking model predictive control scheme for dynamic
  target signals,'' \emph{Automatica}, vol. 118, p. 109030, 2020.

\bibitem{yang_nonlinear_2021}
H.~Yang, H.~Zhao, Y.~Xia, and J.~Zhang, ``Nonlinear {MPC} with time-varying
  terminal cost for tracking unreachable periodic references,''
  \emph{Automatica}, vol. 123, p. 109337, 2021.

\bibitem{Limon_JPC_2014}
D.~Limon, M.~Pereira, D.~M. de~la Pe{\~n}a, T.~Alamo, and J.~Grosso,
  ``Single-layer economic model predictive control for periodic operation,''
  \emph{Journal of Process Control}, vol.~24, no.~8, pp. 1207--1224, 2014.

\bibitem{Pereira_PEMPC_2015}
M.~Pereira, D.~Limon, D.~M. de~la Pe{\~n}a, L.~Valverde, and T.~Alamo,
  ``Periodic economic control of a nonisolated microgrid,'' \emph{IEEE
  Transactions on Industrial Electronics}, vol.~62, no.~8, pp. 5247--5255,
  2015.

\bibitem{Kohler_IFAC_2023}
M.~K{\"o}hler, M.~A. M{\"u}ller, and F.~Allg{\"o}wer, ``Distributed model
  predictive control for periodic cooperation of multi-agent systems,''
  \emph{IFAC-PapersOnLine}, vol.~56, no.~2, pp. 3158--3163, 2023.

\bibitem{krupa_harmonic_2022}
P.~Krupa, D.~Limon, and T.~Alamo, ``Harmonic based model predictive control for
  set-point tracking,'' \emph{IEEE Transactions on Automatic Control}, vol.~67,
  no.~1, p. 48–62, 2022.

\bibitem{krupa_efficiently_2022}
P.~Krupa, D.~Limon, A.~Bemporad, and T.~Alamo, ``Efficiently solving the
  harmonic model predictive control formulation,'' \emph{IEEE Transactions on
  Automatic Control}, vol.~68, no.~9, pp. 5568--5575, 2023.

\bibitem{Karamanakos_OJIA_2020}
P.~Karamanakos, E.~Liegmann, T.~Geyer, and R.~Kennel, ``Model predictive
  control of power electronic systems: Methods, results, and challenges,''
  \emph{IEEE Open Journal of Industry Applications}, vol.~1, pp. 95--114, 2020.

\bibitem{Ordonez_TPE_2023}
J.~G. Ordonez, P.~Montero-Robina, D.~Limon, and F.~Gordillo, ``Real-time
  implementation of predictive control in power inverters based on nearest
  neighbor searching,'' \emph{IEEE Transactions on Power Electronics}, vol.~39,
  no.~1, pp. 384--397, 2024.

\bibitem{Lafmejani_RAL_2020}
A.~S. Lafmejani, A.~Doroudchi, H.~Farivarnejad, X.~He, D.~Aukes, M.~M. Peet,
  H.~Marvi, R.~E. Fisher, and S.~Berman, ``Kinematic modeling and trajectory
  tracking control of an octopus-inspired hyper-redundant robot,'' \emph{IEEE
  Robotics and Automation Letters}, vol.~5, no.~2, pp. 3460--3467, 2020.

\bibitem{dong_novel_2022}
K.~Dong, J.~Luo, and D.~Limon, ``A novel stable and safe model predictive
  control framework for autonomous rendezvous and docking with a tumbling
  target,'' \emph{Acta Astronautica}, vol. 200, p. 176–187, 2022.

\bibitem{krupa_spcies_2020}
P.~Krupa, V.~Gracia, D.~Limon, and T.~Alamo, ``{SPCIES}: {S}uite of predictive
  controllers for industrial embedded systems,''
  https://github.com/GepocUS/Spcies, 2020.

\bibitem{Pereira_Obstacles_2021}
J.~C. Pereira, V.~J. Leite, and G.~V. Raffo, ``An ellipsoidal-polytopic based
  approach for aggressive navigation using nonlinear model predictive
  control,'' in \emph{2021 International Conference on Unmanned Aircraft
  Systems (ICUAS)}.\hskip 1em plus 0.5em minus 0.4em\relax IEEE, 2021, pp.
  827--835.

\bibitem{dosSantos_AUT_2024}
M.~A. dos Santos, A.~Ferramosca, and G.~V. Raffo, ``Set-point tracking {MPC}
  with avoidance features,'' \emph{Automatica}, vol. 159, p. 111390, 2024.

\bibitem{kohler_TAC_2020}
J.~K{\"o}hler, M.~A. M{\"u}ller, and F.~Allg{\"o}wer, ``A nonlinear model
  predictive control framework using reference generic terminal ingredients,''
  \emph{IEEE Transactions on Automatic Control}, vol.~65, no.~8, p.
  3576–3583, 2020.

\bibitem{krupa_MPCT_implem_2021}
P.~Krupa, I.~Alvarado, D.~Limon, and T.~Alamo, ``Implementation of model
  predictive control for tracking in embedded systems using a sparse extended
  {ADMM} algorithm,'' \emph{IEEE Transactions on Control Systems Technology},
  vol.~30, no.~4, pp. 1798--1805, 2021.

\bibitem{krupa_ellipHMPC_arXiv_2023_v1}
P.~Krupa, D.~Limon, A.~Bemporad, and T.~Alamo, ``Harmonic model predictive
  control for tracking periodic references,'' \emph{arXiv preprint,
  arXiv:2310.16723v1}, 2023.

\bibitem{boyd_distributed_2010}
S.~Boyd, N.~Parikh, E.~Chu, B.~Peleato, and J.~Eckstein, ``Distributed
  optimization and statistical learning via the alternating direction method of
  multipliers,'' \emph{Foundations and Trends in Machine Learning}, vol.~3,
  no.~1, pp. 1--122, 2011.

\bibitem{krupa_implementation_2021}
P.~Krupa, D.~Limon, and T.~Alamo, ``Implementation of model predictive control
  in programmable logic controllers,'' \emph{IEEE Transactions on Control
  Systems Technology}, vol.~29, no.~3, p. 1117–1130, 2021.

\bibitem{ODonoghue_SCS_21}
B.~O'Donoghue, ``Operator splitting for a homogeneous embedding of the linear
  complementarity problem,'' \emph{{SIAM} Journal on Optimization}, vol.~31,
  pp. 1999--2023, 2021.

\bibitem{stellato_osqp_2020}
B.~Stellato, G.~Banjac, P.~Goulart, A.~Bemporad, and S.~Boyd, ``{OSQP}: {An}
  operator splitting solver for quadratic programs,'' \emph{Mathematical
  Programming Computation}, vol.~12, no.~4, p. 637–672, 2020.

\end{thebibliography}

\end{document}